\documentclass[draft]{IEEEtran}
\onecolumn

\usepackage{hyperref}
\usepackage[utf8]{inputenc}
\usepackage{amsmath,amssymb,amsthm}
\usepackage{color}
\newtheorem{definition}{Definition}
\newtheorem{theorem}{Theorem}
\newtheorem{corollary}{Corollary}
\newtheorem{lemma}{Lemma}
\newtheorem{proposition}{Proposition}

\newtheorem{remark}{Remark}
\newtheorem{conjecture}{Conjecture}

\DeclareMathOperator{\cov}{cov}
\DeclareMathOperator{\vol}{vol}
\DeclareMathOperator{\var}{var}
\DeclareMathOperator{\tr}{tr}
\DeclareMathOperator{\grad}{grad}
\DeclareMathOperator{\hess}{Hess}
\DeclareMathOperator*{\argmax}{arg\,max}

\title{Stability of the Gaussian Stationary Point in  the  Han-Kobayashi Region for Z-Interference Channels}
\author{%
  \IEEEauthorblockN{Jingbo Liu}\\
  \IEEEauthorblockA{Department of Statistics,
   University of Illinois, Urbana-Champaign\\
                                      Email: jingbol@illinois.edu}
}

\date{December 2021}

\begin{document}

\maketitle

\begin{abstract}
The Gaussian stationary point in an inequality motivated by the Z-interference channel was recently conjectured by Costa, Nair, Ng, and Wang to be the global optimizer, which, if true, would imply the optimality of the Han-Kobayashi region for the Gaussian Z-interference channel.
This conjecture was known to be true for some parameter regimes, but the validity for all parameters, although suggested by Gaussian tensorization, was previously open.
In this paper we construct several counterexamples showing that this conjecture may fail in certain regimes:
A simple construction without Hermite polynomial perturbation is proposed,
where distributions far from Gaussian are analytically shown to be better than the Gaussian stationary point.
As alternatives,
we consider perturbation along geodesics under either the $L^2$ or Wasserstein-2 metric,
showing that the Gaussian stationary point is unstable in a certain regime.
Similarity to stability of the Levy-Cramer theorem is discussed.
The stability phase transition point admits a simple characterization in terms of the maximum eigenvalue of the Gaussian maximizer. 
Similar to the Holley-Stroock principle,
we can show that in the stable regime the Gaussian stationary point is optimal in a neighborhood under the $L^{\infty}$-norm with respect to the Gaussian measure.
For protocols with constant power control,
our counterexamples imply Gaussian suboptimality for the Han-Kobayashi region.
Allowing variable power control, we show that the Gaussian optimizers for the Han-Kobayashi region always lie in the stable regime.
We propose an amended conjecture, whose validity would imply Gaussian optimality of the Han-Kobayashi bound in a certain regime.
\end{abstract}

\section{Introduction}
Interference channel is a fundamental problem in multiuser information theory, whose single-letter rate region has remained open after decades of efforts.
A two-user interference channel consists of two inputs $X_1$ and $X_2$ and two outputs $Y_1$ and $Y_2$.
In general, both $Y_1$ and $Y_2$ are noisy versions of functions of $(X_1,X_2)$, and the goal is to find the maximum transmission rates $(R_1,R_2)$ achievable by the two users.
For a more precise formulation of the problem, see e.g.
\cite{costa1985gaussian}\cite{elgamal_kim_2011}.
In the special case of Gaussian Z-interference channel, $Y_1$ is a noisy version only of $X_1$ rather than of $(X_1,X_2)$. 
With some transformations, the channel model can be expressed as (see \cite{costa2020structure})
\begin{align}
Y_1&=X_1+Z_1;\label{e_y1}
\\
Y_2&=X_2+X_1+Z_1+Z_2,\label{e_y2}
\end{align}
where $Z_i\sim \mathcal{N}(0,N_iI)$, $i=1,2$, and $X_1$, $X_2$, $Z_1$, $Z_2$ are independent.
The Gaussian Z-interference channel is significant, partly because it is equivalent to several versions of  degraded Gaussian interference channels \cite{costa1985gaussian}.

Due to the sheer volume, we do not attempt to provide here a comprehensive review of the all the important work on the interference channel since the 70's.
Interested readers may refer to the survey \cite{shang2013two}
or the more recent papers \cite{beigi2016some}\cite{costa2020structure}\cite{gohari2021information}.
Below, we cite a few properties of the interference channel relevant to our purpose:
\begin{itemize}
\item The multi-letter Han-Kobayashi (HK) inner bound \cite{han1981new} is tight, whereas the single letter Han-Kobayashi bound is known to be suboptimal for some discrete channels where tensorization fails \cite{nair2015sub}.
\item When restricted to Gaussian inputs, the HK-region is known to satisfy the tensorization property \cite{nair2018invariance}.
\item It is unknown whether Gaussian inputs are optimal for the HK bound  for Gaussian Z-interference channels.
If optimal, this would resolve the longstanding open problem about the rate region of the Gaussian Z-interference channel, because of the two itemized properties above. 
More or less motivated by this, the Gaussian optimality problem spurred a lot of research interests recently, e.g.\
\cite{beigi2016some}\cite{courtade2017strong}\cite{liu2018forward}\cite{costa2020structure}\cite{ng2021gaussian}\cite{gohari2021information}\cite{anantharam2022unifying}.
It was known that Gaussian inputs are optimal for computing the corners of the region \cite{sato1981capacity}\cite{costa1985gaussian}\cite{polyanskiy2016wasserstein}\cite{beigi2016some}\cite{costa2020structure}\cite{gohari2020new}\cite{gohari2021information}\cite{gohari2021outer},
but the full region or even the precise slope at Costa's corner point remains open \cite{gohari2021information}.
\item Gaussian inputs are known to be suboptimal for the symmetric Gaussian interference channel under constant power control \cite{abbe2012coordinate} for some range of parameters, where the argument was based on perturbations using Hermite polynomials. 
Some care needs to be taken in choosing how to perturb and ensuring that the perturbed density is still a probability measure.
\item Constant power control is  suboptimal for the class of schemes using Gaussian inputs, known as ``Gaussian noisebergs'' \cite{costa2011noisebergs} (see also  \cite{shang2013two}).
\end{itemize}
The second itemized above (Gaussian tensorization) seems to suggest that Gaussian inputs are optimal for the Han-Kobayashi region. 
Indeed, a neat rotation-invariance argument usually succeeds in showing Gaussian optimality in entropic inequalities whenever there is a tensorization property \cite{geng2014capacity} (see also similar arguments in the context of functional inequalities \cite{lieb1990gaussian}\cite{liu2018forward}).
However, the tensorization property in \cite{nair2018invariance} is for Gaussian inputs rather than general inputs, therefore Gaussian optimality cannot be settled in this manner.

Since the HK bound is notoriously complicated, it is useful to formulate simpler necessary and sufficient conditions for checking its tightness.
Recently, Costa, Nair, Ng and Wang \cite{costa2020structure} 
(see also the related \cite{beigi2016some}\cite{ng2021gaussian}\cite{gohari2021information}) proposed the following conjecture about linear combination of differential entropies which, if true, would imply Gaussian optimality in the Han-Kobayashi region.

\begin{conjecture}[\cite{costa2020structure}]\label{conj1}
For $u\ge 0$, $N_1,N_2\ge 0$, and $\Sigma_1,A_2\succeq0$, the maximum 
\begin{align}
\max_{
\substack{P_{X_1}P_{X_2}\\ \mathbb{E}[X_2X_2^{\top}]\preceq A_2}
}
\left\{
u h(X_1+X_2+Z_1+Z_2)
+h(X_1+Z_1)\right.
\nonumber\\
-\left.(1+u) h(X_1+Z_1+Z_2)-\tr(\Sigma_1\mathbb{E}[X_1X_1^{\top}])
\right\}\label{e_1}
\end{align}
where 
$Z_1\sim\mathcal{N}(0,N_1I)$, $Z_2\sim \mathcal{N}(0,N_2I)$ and $X_i$, $Z_i$ ($i=1,2$) are random variables in $\mathbb{R}^d$ ($d\ge 1$), is attained by Gaussian $X_1$ and $X_2$.
\end{conjecture}
It is easy to see that if $(P_{X_1},P_{X_2})$ is the maximizer among all Gaussian distributions, then it is also a stationary point (i.e.\ vanishing gradient) in the space of all probability distributions (under suitable metric and regularity condition). 
Indeed, with a perturbation of a Gaussian density, the first-order change of the differential entropy depends only on the change of the covariance (see e.g.~\eqref{e_ent}),
so the stationarity condition is the same as $(P_{X_1},P_{X_2})$ having the optimal covariance.

Indications that the conjecture might be true include the aforementioned Gaussian tensorization property, and the fact that if either $X_1$ or $X_2$ is Gaussian then the best choice of the other random variable is also Gaussian \cite{beigi2016some}\cite{costa2020structure}.
Previously, the conjecture has been proved for large enough $u$, which is enough to evaluate the corner points of the rate region, even though these corner points were previously established by other methods; see \cite{gohari2021information}.
However, it was not known whether it is true for all $u$, so that one can deduce Gaussian optimality for the capacity \emph{region}, which is the central open problem on Z-interference channels.
As a partial converse, \cite{beigi2016some} observed that disproving Conjecture~\ref{conj1} around Costa's corner point would imply suboptimality of the HK region.

In this paper, we show that Conjecture~\ref{conj1} may fail for some range of parameters, by constructing counterexamples with several methods:
\begin{itemize}
\item We choose a set of parameters for which \eqref{e_1} is nonpositive with Gaussian inputs, yet strictly positive for some non-Gaussian inputs (Section~\ref{sec2}).
This simple proof is very  different from the existing method of Hermite polynomial perturbation \cite{abbe2012coordinate}.
In fact, our non-Gaussian inputs have non-vanishing non-Gaussianness in the proof. 
We discuss an interesting analogous phenomenon in a geometric inequality (deferred to Section~\ref{sec_geometric}).
\item Alternatively, a counterexample can be constructed by perturbing the Gaussian  ``vertically'', i.e.\ along a geodesic under the $L^2(\mathbb{R})$ metric (Section~\ref{sec_vertical}).
This amounts to computing the Hessian under $L^2(\mathbb{R})$,
revealing that the Gaussian stationary point is unstable in a certain parameter regime.
Since differentiation commutes with convolution, 
derivatives of the Gaussian density are useful in constructing the direction of perturbation.
which is closely related the Hermite perturbation approach \cite{abbe2012coordinate}.
Previously, \cite{nair2018scalar} commented that attempts of disproving Conjecture~\ref{conj1} using the method of \cite{abbe2012coordinate} were unsuccessful.
We note that in order to ensure nonnegativity of the density after perturbation, 
\cite{abbe2012coordinate} introduced a device of adding an additional even-order Hermite polynomial (because odd-order Hermite polynomials are unbounded in both directions, whereas even-order Hermite polynomials are bounded below); 
here we handle the nonnegativity issue in another way which is simpler for the present problem (Remark~\ref{rem_nonnegative}).
To explain the intuition we observe the similarity to the stability problem in the Levy-Cramer theorem (Remark~\ref{rem_lc}).
\item A counterexample can also be constructed by computing the Hessian under the Wasserstein-2 distance, based on the Otto-Villani framework of viewing the space of probability measures as an infinite dimensional Riemannian manifold with the Wasserstein-2 metric (Section~\ref{sec_horizontal}).
In standard Riemannian geometry, the formula of the Hessian in a coordinate system depends on the second derivative as well as the product of the connection and the first derivative (see e.g.\ \cite[(2.6)]{gao2018semi}),
from which one can read off how the Hessian changes under a change of metric.
At a general point on the manifold, 
semidefiniteness of the Hessian depends on the choice of metric, because of the term of the product of the connection and the first derivative.
However at a stationary point where the first derivative vanishes,
semidefiniteness of the Hessian does not depend on the choice of metric.
In Section~\ref{sec_horizontal} we provide self-consistent calculations confirming that the stability phase transition point of the Hessian is consistent under either the $L^2$ or Wasserstein-2 distance.
In contrast to the vertical perturbation method, with the Wasserstein flow approach we never need to worry about the nonnegativity constraint of the density.
Moreover, recently, first and second order gradient descent methods under the Wasserstein metric have been used in numerically solving optimization problems, 
with convenient particle flow interpretation and computational advantages over the traditional descent methods under the $L^2$ metric \cite{ambrosio2008gradient}\cite{craig2016exponential}\cite{wang2020information}\cite{chewi2020gradient}.
\end{itemize}
For parameters of Conjecture~\ref{conj1} giving rise to negative-definite Hessian at the Gaussian stationary point,
we show that the Gaussian stationary point is stable in the sense that it is a local maximizer under the $L^{\infty}$-norm with respect to the Gaussian reference measure (Section~\ref{sec_local}).
The choice of $L^{\infty}$-norm is in the same spirit as the celebrated Holley-Stroock principle \cite[p1184]{holley}\cite{ledoux2001logarithmic}, which is a method of proving functional inequalities (such as log-Sobolev or Poincare) under a perturbation.
For such a local optimality result, it suffices to bound $h(X_1+X_2+Z_1+Z_2)$ in \eqref{e_1} by the surrogate $\frac1{2}\ln[(2\pi e)^d\det(\cov(X_1+X_2+Z_1+Z_2))]$.
It is well-known in variational calculus that the choice of topology is critical for local optimality, and in fact this is precisely the case here: 
local optimality (with the surrogate objective function) is not true under the $L^{\infty}$-norm with respect to the \emph{Lebesgue measure}; 
a similar phenomenon was previous observed concerning the stability of the log-Sobolev inequality \cite[Section~6]{eldan2020stability}.

From the Hessian computations we see that the stability phase
transition point admits a simple characterization in terms of the maximum eigenvalue of the covariance matrix of the Gaussian maximizer in \eqref{e_1}.
On the other hand, in Section~\ref{sec_structure} we establish a bound on the maximum eigenvalue of the covariance matrix of the  Gaussian maximizer in the HK region with (time-varying) power control.
Comparing the two results,
we see that our counterexamples are not sufficient for disproving Gaussian optimality for the HK region with power control, nor do they prove \cite[Hypothesis~1]{beigi2016some}.
Nevertheless, in Section~\ref{sec_application}, we show that the counterexamples are sufficient for showing Gaussian suboptimality among protocols with constant power control, 
which is a new result.
Previously, Gaussian with constant power control was shown to be suboptimal among protocols allowing power control \cite{costa1985gaussian}.

While stability of stationary points is not sufficient for global optimality, 
it is often the first step towards understanding many phase transition problems (e.g. \cite{de1978stability}).
Our results suggest that towards the grand goal of settling Gaussian optimality in the HK region with power control,
we need to modify Conjecture~\ref{conj1}, or find counterexamples not based on Gaussian perturbation (under the $L^{\infty}$-norm).
For a compact version that seems to represent much of the challenge,  
we may consider the following limiting special case of the optimization problem:
(see Remark~\ref{rem9}, Theorem~\ref{thm_2}, and Lemma~\ref{lem5} for explanations of the connection):
\begin{align}
\sup\left\{h(X+Y)-h(X)-\frac1{2}J(X)\right\}
\label{e2}
\end{align}
where $X$ and $Y$ are independent one-dimensional random variables, $\mathbb{E}[Y]\le L$, and $J(\cdot)$ denotes the Fisher information \cite{cover1999elements}.
The best Gaussian $X$ (stationary point) has variance $\frac{L}{L-1}$ when $L>1$.
Results of this paper imply that the Gaussian stationary point is not optimal for $1<L<1.5$, but local optimal (stable) for $L>1.5$. 
The true Gaussian maximizer in the HK bound with power control concerns the region of $L>2$.
If the Gaussian stationary point in \eqref{e2} is not global optimal for some value of $L>2$, we can show that the HK inner bound is not tight.

Finally, in Section~\ref{sec_discussion} we propose an amended conjecture that takes into account of the stability regime of parameters (more precisely, it concerns the regime where power control is not necessary for the Gaussian HK bound).
The new conjecture is nontrivial, since either proving or disproving it would imply the solution to some other open questions in the literature.

{\bf Notation.} We denote the differential entropy \cite{cover1999elements} of a random variable $X\sim P_X$ by $h(X)$ or $h(P_X)$.
The unit of entropy is nat.
$\gamma_K$ denotes the centered Gaussian measure with covariance matrix $K$.
When used as a density function, $\gamma_K$ denotes the density with respect to the Lebesgue measure, unless otherwise stated.
The norms with respect to a reference measure are computed using the density with respect to the reference measure; for example, $\|\gamma_K\|_{L^2(P)}^2=\int (\frac{d\gamma_K}{dP})^2 dP$.
We use standard Landau notations such as $\Theta()$, $\omega()$, and $O()$.

\section{Gaussian Suboptimality without Gaussian Perturbation}
\label{sec2}
In this section we construct a counterexample for Conjecture~\ref{conj1}.
We first note the following  about the asymptotic expansion of the differential entropy of convolution, which is similar to the calculations in the de Bruijn's inequality (see e.g.\ \cite{cover1999elements}).
\begin{lemma}\label{lem1}
Let $p$ and $q$ be smooth density functions on $\mathbb{R}$. 
Denote the moments of $q$ by
$m_i:=\int x^iq(x)dx=0$, $i=1,2,\dots$, and suppose that $m_1=0$.
Suppose that there are positive constant $c$ and $C$ such that for any $x\in\mathbb{R}$,
\begin{align}
q(x)&\le Ce^{-cx^2},
\\
p(x)&>c\,e^{-C|x|},
\\
\max\{|p'(x)|,|p'''(x)|,|p''''(x)|\}&<Cp(x).\label{e4}
\end{align}
Let $p_t$ be the convolution of $p$ and $t^{-1/2}q(\frac{\cdot}{\sqrt{t}})$.
Then for $t>0$ small, we have the following estimate for the differential entropy:
\begin{align}
&\quad h(p_t)-h(p)
\nonumber\\
&=m_2(-\frac1{2}\int p''\ln p)t
+m_3(-\frac1{6}\int p'''\ln p)t^{3/2}+O(t^2)
\label{e6}
\end{align}
where $O()$ hides constants that may depend on $p$ and $q$.
\end{lemma}
\begin{proof}
Using the Taylor expansion, we have
\begin{align}
p_t(y)
&=\int p(\sqrt{t}u+y)q(u)du
\\
&=\int\left[p(y)+\sqrt{t}up'(y)+\frac1{2}tu^2p''(y)\right.
\nonumber\\
&\quad \left.+\frac1{6}t^{3/2}u^3p'''(y)+\frac1{24}t^2u^4p''''(\xi_{y,u})\right]q(u)du
\end{align}
where $\xi_{y,u}$ is a value between $y$ and $y+\sqrt{t}u$.
Therefore for any $t<1$,
\begin{align}
&\quad \left|p_t(y)-p(y)-\frac1{2}tm_2p''(y)-\frac1{6}t^{3/2}m_3p'''(y)\right|
\nonumber\\
&=\frac1{24}t^2\left|\int u^4p''''(\xi_{y,u})q(u)du
\right|.
\label{e8}
\\
&\le \frac{t^2}{24}p(y)\int u^4
\left(\sup_{v\colon |v-y|\le \sqrt{t}|u|}\frac{|p''''(v)|}{p(y)}\right)q(u)du
\\
&\le Ct^2p(y)
\int u^4e^{C\sqrt{t}|u|}q(u)du
\label{e7}
\\
&\le C^2t^2p(y)\int 
u^4e^{C|u|-cu^2}du
\label{e9}
\\
&\le c_1t^2p(y)
\label{e_8}
\end{align}
where here and below $c_1,c_2,\dots$ are positive constants depending on $p$ and $q$ but not on $t$.
To see \eqref{e7},
note that \eqref{e4} implies for $|(\ln p(x))'|<C$ and hence
\begin{align}
\sup_{v\colon |v-y|
\le \sqrt{t}|u|}\frac{|p''''(v)|}{p(y)}
&\le
\sup_{v\in\mathbb{R}}\frac{|p''''(v)|}{p(v)}\cdot
\sup_{v\colon |v-y|
\le \sqrt{t}|u|}\frac{p(v)}{p(y)}
\\
&\le Ce^{C\sqrt{t}|u|}.
\end{align} 
Next define $\Delta_t(y):=p_t(y)-p(y)$.
By the Taylor expansion of the function $x\mapsto x\ln x$ around $p(y)$, we have
\begin{align}
&\quad\left|p_t(y)\ln p_t(y)-p(y)\ln p(y)
-\Delta_t(y)-\Delta_t(y)\ln p(y)
\right|
\nonumber\\
&\le \frac{\Delta_t^2(t)}{2\,\xi}
\\
&\le 
\frac1{2}\cdot\frac{\Delta_t^2(y)}{p_t(y)\wedge p(y)}
\label{e16}
\end{align}
where $\xi$ denotes a number between $p(y)$ and $p_t(y)$ (from the Lagrange remainder) and $\wedge$ denotes the minimum value.
Noting $\int \Delta_t=0$, we have
\begin{align}
\left|\int p_t\ln p_t
-\int p\ln p-\int \Delta_t\ln p\right|
\le \int\frac{\Delta_t^2}{p_t\wedge p}.
\label{e10}
\end{align}
By \eqref{e_8}, for $t<1$ we have
\begin{align}
|\Delta_t(t)|&\le 
\frac1{2}tm_2|p''(y)|+\frac1{6}t^{3/2}m_3|p'''(y)|+c_1t^2p(y)
\nonumber\\
&\le c_2tp(y).
\end{align}
Thus for $t<\frac1{2c_2}\wedge1$ we have $\frac{t^{-2}\Delta_t^2(y)}{p_t(y)\wedge p(y)}\le \frac{t^{-2}\Delta_t^2(y)}{\frac1{2}p(y)}\le  2c_2^2p(y)$ and
\begin{align}
\int\frac{\Delta_t^2}{p_t\wedge p}
\le 2c_2^2t^2.
\label{e18}
\end{align}
Moreover, by multiplying $\ln p(y)$ to \eqref{e_8} and integrating, we find that for $t< \frac1{2c_2}\wedge1$,
\begin{align}
\left|\int \Delta_t\ln p-\frac1{2}tm_2\int p''\ln p
-\frac1{6}t^{3/2}m_3\int p'''\ln p\right|
\le 
c_3t^2.
\label{e19}
\end{align}
Comparing \eqref{e10}, \eqref{e18} and \eqref{e19} we establish \eqref{e6}.

\end{proof}

\begin{lemma}\label{lem2}
Suppose that $Z_2\sim \mathcal{N}_2(0,N_2)$, $N_2>0$. Then the supremum of 
\begin{align}
h(X_1+Z_2+X_2)+h(X_1)-2h(X_1+Z_2)
\label{e17}
\end{align}
over the distribution of $(X_1,X_2)$, where $X_1,X_2,Z_2$ are independent and $\mathbb{E}[X_2^2]=N_2$, 
is strictly positive.
Moreover if $X_2$ is restricted to be Gaussian then the supremum equals 0.
\end{lemma}
\begin{proof}
First, if $X_2$ is Gaussian, then
\begin{align}
&\quad h(X_1+Z_2+X_2)-h(X_1+Z_2)
\nonumber\\
&=I(X_1+Z_2+X_2;X_2)
\\
&= I(X_1+Z_2+X_2;Z_2)
\\
&\le I(X_1+Z_2;Z_2)
\\
&=h(X_1+Z_2)-h(X_1),
\end{align}
implying that the supremum is non-positive. Further, taking $X_1$ to be Gaussian with a large variance we see that the supremum equals 0.

Now consider the case where $X_2$ is not restricted to be Gaussian.
Choose $p$ and $q$ satisfying the conditions in Lemma~\ref{lem1} and such that
\begin{align}
\frac1{6}\int
p'''\ln p&>0;
\\
m_1&=0;
\\
m_2&>0;
\\
m_3&<0.
\end{align}
It should be clear that the supremum of \eqref{e17} is independent of the choice of $N_2$ (by considering a scaling of the random variables).
For simplicity let us assume that $N_2=m_2$ below.
Let $X_2\sim q$ and let $X_1$ be such that $\sqrt{t}X_1\sim p$ (in particular, observe that the distribution of $X_1$ depends on $t$, but the distribution of $\sqrt{t}X_1$ does not).
Note that 
\begin{align}
\mathbb{E}[(Z_2+X_2)^2]&=N_2+m_2=2m_2;
\\
\mathbb{E}[(Z_2+X_2)^3]&=\mathbb{E}[X_2^3]=m_3.
\end{align}
Then
\begin{align}
&\quad h(X_1+Z_2+X_2)-h(X_1)
\nonumber\\
&=
h(\sqrt{t}X_1+\sqrt{t}Z_2+\sqrt{t}X_2)-h(\sqrt{t}X_1)
\\
&=2m_2(-\frac1{2}\int p''\ln p)t
+m_3(-\frac1{6}\int p'''\ln p)t^{3/2}
\nonumber\\
&\quad+O(t^2);
\label{e22}
\end{align}
Similarly, since $Z_2$ is symmetric whose third moment equals zero, we have
\begin{align}
&\quad h(X_1+Z_2)-h(X_1)
\nonumber\\
&=h(\sqrt{t}X_1+\sqrt{t}Z_2)-h(\sqrt{t}X_1)
\\
&=
m_2(-\frac1{2}\int p''\ln p)t
+O(t^2).
\label{e23}
\end{align}
Then we note that \eqref{e17} equals \eqref{e22} minus twice of \eqref{e23}, which is $m_3(-\frac1{6}\int p'''\ln p)t^{3/2}+O(t^2)$.
Thus the claim about positivity of the supremum follows when $t$ is sufficiently small.
\end{proof}
\begin{remark}
Although not necessary for establishing a counterexample to Conjecture~\ref{conj1}, let us comment that $X_2$ being Gaussian implies that Gaussian $X_1$ is optimal.
Indeed, this has been shown using the doubling trick \cite{costa2020structure}, or, even simpler, using Costa's entropy power inequality \cite{beigi2016some}.
\end{remark}

Clearly Lemma~\ref{lem2} provides a counterexample for Conjecture~\ref{conj1} with
$u=1$, $N_1=0$, $N_2=A_2$, and $\Sigma_1=0$. 
While strictly positive $N_1$ sometimes plays a role in analysis related to interference channels, 
for example in the proof of Wasserstein continuity of smoothed entropy, 
it is not essential in the proof of Gaussian suboptimality.
Below, we show that a counterexample exists for $N_1$ and $\Sigma_1$ strictly positive as well:
\begin{theorem}\label{thm_continuity}
Let $k=u=N_1=1$. There exists some $N_2=A_2>0$ and $\Sigma_1>0$ for which \eqref{e_1} is strictly positive yet the supremum restricted to Gaussian $X_2$ is not positive.
\end{theorem}
\begin{proof}
We first provide a proof using a continuity argument. By replacing $X_1$ in Lemma~\ref{lem2} with $X_1+Z_1$, $Z_1\sim \mathcal{N}(0,N_1)$, and by continuity of the differential entropy in $N_1$ (see e.g.\ \cite[Lemma A.3.]{anantharam2022unifying}),
we see that the supremum of
\begin{align}
h(X_1+Z_1+Z_2+X_2)+h(X_1+Z_1)-2h(X_1+Z_1+Z_2)-\Sigma_1\mathbb{E}[X_1^2]
\label{e25}
\end{align}
over the distribution of $(X_1,X_2)$, where $X_1,X_2,Z_2$ are independent and $\mathbb{E}[X_2^2]\le N_2=\var(Z_2)$, 
is strictly positive provided that $N_1$ and $\Sigma_1$ are sufficiently small.
Also, it is clear that the supremum in \eqref{e25} does not change under the transformations $N_1'=1$, $A_2'=N_2'=\frac{N_2}{N_1}$, and $\Sigma_1'=\Sigma_1N_1$ (to see this consider $X_1'=\frac1{\sqrt{N_1}}X_1$ and $X_2'=\frac1{\sqrt{N_1}}X_2$).
Therefore, equivalently, fixing $N_1=1$, we can always find $N_2=A_2>0$ and $\Sigma_1>0$ so that the supremum is positive.
In the meantime, if $X_2$ is restricted to be Gaussian then the supremum does not exceed 0 as in Lemma~\ref{lem2}.

An alternative argument for positivity of \eqref{e25} was suggested by Chandra Nair:
for $X_1\sim p$ as in Lemma~\ref{lem1}, we can pick a finite $N_1>0$ and verify that $X_1+Z_1$ still satisfies the assumptions on $p$ in Lemma~\ref{lem1}.
Then treating $X_1+Z_1$ as the new $X_1$, the proof of Lemma~\ref{lem2} still shows that \eqref{e25} is positive for some $\Sigma_1>0$, $N_2=A_2>0$ and $X_2$.  
\end{proof}

\section{Gaussian Suboptimality via Vertical Perturbation}\label{sec_vertical}
In this section we provide an alternative construction of counterexample to Conjecture~\ref{conj1} via ``vertical perturbation''\footnote{In the literature, the name ``vertical perturbation'' usually refers to a perturbation of the probability density, instead of perturbing along a Wasserstein geodesic (c.f.\ Section~\ref{sec_horizontal}).} of Gaussian. 
First, let us observe the following about the differential entropy under the vertical perturbation.

\begin{proposition}
Assume that $p$ is probability density function on $\mathbb{R}^d$, and $r$ is a measurable function on $\mathbb{R}^d$ satisfying $\sup_{x\in\mathbb{R}^d}|\frac{r(x)}{p(x)}|<\infty$ and $\int r=0$.
Then as $\epsilon\to0$, we have
\begin{align}
\int (p+\epsilon r)
\ln(p+\epsilon r)
=\int p\ln p
+\epsilon\int r\ln p
+\frac{\epsilon^2}{2}\int\frac{r^2}{p}+O(\epsilon^3).
\label{e_ent}
\end{align}
\end{proposition}
\begin{proof}
Using the Taylor expansion of the function $t\mapsto t\ln t$ we see that for $\epsilon<\inf_{x\in\mathbb{R}^d}|\frac{p(x)}{r(x)}|$,
\begin{align}
(p+\epsilon r)\ln (p+\epsilon r)=p\ln p+\epsilon r+\epsilon r\ln p+\frac{\epsilon^2r^2}{2\xi}
\label{eq_41}
\end{align}
where the function $\xi(x)$ is between $p(x)$ and $p(x)+\epsilon r(x)$.
For $\epsilon<\frac1{2}\inf_{x\in\mathbb{R}^d}|\frac{p(x)}{r(x)}|$ we have
\begin{align}
\left|\frac{r^2}{\xi}-\frac{r^2}{p}\right|
=\frac{r^2|p-\xi|}{\xi p}\le \frac{2\epsilon r^3}{p^2}.
\end{align}
The claim then follows by integrating \eqref{eq_41}.
\end{proof}
The key result in the construction of counterexample is the following:
\begin{theorem}\label{thm_2}
For any $u>0$, $L>0$, and $T>0$,
let $K>\frac{u}{(1+u)^{1/3}-1}$.
There exist $P_{X_1}^{\epsilon}$ and $P_{X_2}^{\epsilon}$ (indexed by $\epsilon$ in a neighborhood of 0) such that
$h(P_{X_1}^{\epsilon}*\gamma_u*P_{X_2}^{\epsilon})=h(\gamma_{K+u+L})+O(\epsilon^{2T})$ and $h(P_{X_1}^{\epsilon})-(u+1)h(P_{X_1}^{\epsilon}*\gamma_u)
\ge h(\gamma_{K})-(u+1)h(\gamma_{K+u})+A\epsilon^2$ for some $A>0$, as $\varepsilon\to0$.
\end{theorem}

\begin{proof}
Define 
\begin{align}
P_{X_1}^{\epsilon}=&\gamma_{K}-\epsilon D^3\gamma_{K-\delta};
\label{e48}
\\
P_{X_2}^{\epsilon}=&\sum_{j=0}^{J}
\epsilon^{j}D^{3j}\gamma_{L-j\delta},
\label{e49}
\end{align}
where the positive integer $J$ satisfies $J+1\ge T$, and $\delta>0$ satisfies $K-\delta>0$, $L-J\delta>0$.
Moreover $D^3\gamma_{K-\delta}$ means differentiating the Gaussian density 3 times, which yields a finite measure on $\mathbb{R}$.
Thus \eqref{e48}-\eqref{e49} define valid probability measures when $\epsilon$ is sufficiently close to 0.

Then for $Z_2\sim\mathcal{N}(0,u)$ we have 
\begin{align}
P_{X_1}^{\epsilon}*\gamma_u&=
(\gamma_{K}-\epsilon D^3\gamma_{K-\delta})
*\gamma_{u}
\\
&=\gamma_{K+u}-\epsilon D^3\gamma_{K+u-\delta}
\end{align}
and 
\begin{align}
P_{X_1}^{\epsilon}*P_{X_2}^{\epsilon}
&=
(\gamma_{K}-\epsilon D^3\gamma_{K-\delta})
*
(\sum_{j=0}^J
\epsilon^{j}D^{3j}\gamma_{L-j\delta})
\\
&=\sum_{j=0}^J
\epsilon^{j}D^{3j}\gamma_{K+L-j\delta}
-
\sum_{j=0}^J
\epsilon^{j+1}D^{3(j+1)}\gamma_{K+L-(j+1)\delta}
\\
&=\gamma_{K+L}
-\epsilon^{J+1}D^{3(J+1)}\gamma_{K+L-(J+1)\delta}.
\end{align}
Therefore
\begin{align}
P_{X_1}^{\epsilon}*\gamma_u*P_{X_2}^{\epsilon}
=\gamma_{K+u+L}
-\epsilon^{J+1}D^{3(J+1)}\gamma_{K+u+L-(J+1)\delta}.
\end{align}
Using \eqref{e_ent} we have
\begin{align}
h(P_{X_1}^{\epsilon})=h(\gamma_K)-\frac{\epsilon^2}{2}\int\frac{(D^3\gamma_{K-\delta})^2}{\gamma_K}
+o(\epsilon^2);
\end{align}
\begin{align}
h(P_{X_1}^{\epsilon}*\gamma_u)=h(\gamma_{K+u})
-\frac{\epsilon^2}{2}\int\frac{(D^3\gamma_{K+u-\delta})^2}{\gamma_{K+u}}
+o(\epsilon^2);
\end{align}
and 
\begin{align}
h(P_{X_1}^{\epsilon}
*\gamma_u*P_{X_2}^{\epsilon})
=h(\gamma_{K+u+L})-\frac{\epsilon^{2(J+1)}}{2}\int\frac{(D^{3(J+1)}\gamma_{K+u+L-(J+1)\delta})^2}{\gamma_{K+u+L}}
+o(\epsilon^{2(J+1)}).
\end{align}
Clearly the claim of the theorem holds if 
\begin{align}
-\int\frac{(D^3\gamma_{K-\delta})^2}{\gamma_K}
+(1+u)\int\frac{(D^3\gamma_{K+u-\delta})^2}{\gamma_{K+u}}>0.
\label{e54}
\end{align}
By the property of the Hermite polynomial (see e.g.\ \cite{cramer1999mathematical}) we have 
\begin{align}
\int\frac{(D^k\gamma_K)^2}{\gamma_K}
=
\frac{1}{K^k}\int \frac{(D^k\gamma)^2}{\gamma}=\frac{1}{K^k}k!
\label{e_herm}
\end{align}
for any positive integer $k$.
Thus by dominated convergence we see that as $\delta\to0$ the left side of \eqref{e54} converges to
$
-\frac{1}{K^3}k!
+\frac{1+u}{(K+u)^3}k!
$.
Therefore \eqref{e54} holds under $K>\frac{u}{(1+u)^{1/3}-1}$.
\end{proof}
\begin{remark}\label{rem_nonnegative}
Setting $\delta>0$ in \eqref{e48} ensures that it defines a valid probability measure with nonnegative density for small $\epsilon$. 
Note that if $\delta=0$ then the density of the perturbation part (with respect to $\gamma_K$) is a Hermite polynomial of order three.
Previously, Hermite polynomial perturbation has been used in \cite{abbe2012coordinate} for proofs of Gaussian suboptimality.
In order to ensure nonnegative density, \cite{abbe2012coordinate} would add yet another even degree Hermite polynomial (multiplied by $\gamma_K$) to \eqref{e48}, since an even degree Hermite polynomial is bounded from below.
The ratio of the odd and even degree polynomials needs be selected with care, and a similar issue for the perturbation of $X_2$ would add more restrictions.
Here we adopt a simpler trick without adding an even degree polynomial, but choosing a smaller Gaussian variance.
\end{remark}

\begin{remark}\label{rem_lc}
The intuition behind Theorem~\ref{thm_2} is closely related to the Levy-Cramer theorem, which says that the sum of two independent non-Gaussian random variables cannot be precisely Gaussian \cite{cramer1936eigenschaft}.
However, the non-Gaussianness of the sum can be much smaller than that of the individual summands;
Theorem~\ref{thm_2} essentially showed this where the non-Gaussianness is gauged by the Gaussian-regularized relative entropy.
The reverse direction (lower bounding the non-Gaussianness of the sum using non-Gaussianness of the summands) is the problem of stability of the Levy-Cramer theorem, the regularized entropy version of which was considered in \cite{bobkov2016regularized}.
\end{remark}
\begin{corollary}\label{cor1}
Consider any $u>0$ and 
$L>1$ satisfying
$\frac{L+u}{L-1}>\frac{u}{(1+u)^{1/3}-1}$.
There exists some $(P_{X_1},P_{X_2})$ for which 
\begin{align}
&\quad uh(P_{X_1}*\gamma_u*P_{X_2})+h(P_{X_1})
-(u+1)h(P_{X_1}*\gamma_u)
\nonumber\\
&>
\sup_{K>0}\left\{uh(\gamma_{K+u+L})
+h(\gamma_K)-(u+1)h(\gamma_{K+u})\right\}.
\label{e61}
\end{align}
\end{corollary}
\begin{proof}
We can check that the supremum in \eqref{e61} is achieved at  $K=\frac{L+u}{L-1}$ if $L>1$.
Under the assumption of the corollary the claim follows from Theorem~\ref{thm_2}.
\end{proof}

Corollary~\ref{cor1} provides a counterexample to Conjecture~\ref{conj1} for the case of $N_1=\Sigma_1=0$ and $N_2=u$.
Extension to the $N_1,\Sigma_1>0$ case is possible using a similar continuity argument as Theorem~\ref{thm_continuity}).
Alternatively, we can pick finite but small enough $N_1>0$, define
$
P_{X_1}^{\epsilon}=\gamma_{K-N_1}-\epsilon D^3\gamma_{K-N_1-\delta}$,
and then the argument in Theorem~\ref{thm_2} still carries through with $X_1$ replaced by $X_1+Z_1$.

\section{Gaussian Suboptimality via Wasserstein Flow}\label{sec_horizontal} 
In this section we consider another approach of establishing Gaussian suboptimality by showing that the Hessian under the Wasserstein-2 metric at the stationary point fails to be negative semidefinite.
The choice of metric is in fact immaterial to the stability of a stationary point, as we shall explain.
However, it seems to be an interesting direction for future research to leverage the recent developments on optimization in the Wasserstein space 
\cite{ambrosio2008gradient}\cite{craig2016exponential}\cite{wang2020information}\cite{chewi2020gradient}.
Let us also remark that optimal transportation has previously been used in proving Costa's corner point for the interference channels \cite{polyanskiy2016wasserstein}.

\subsection{Preparations}
It is useful to first recall how Hessian is related to stability at a stationary point in the simpler setting of optimization in the Euclidean space.
Suppose that the goal is to maximize
$
f(x)
$ 
over $x\in\mathbb{R}^d$ subject to the constraint $g(x)\le 0$,
where $f$ and $g$ are both smooth functions on $\mathbb{R}^d$.
Then $x^*\in\mathbb{R}^d$ is a stationary point of the constrained optimization problem
if and only if there exists $\lambda\ge 0$ such that $g(x^*)=0$ and 
\begin{align}
\nabla f(x^*)-\lambda \nabla g(x^*)=0\in\mathbb{R}^d.
\label{e72}
\end{align}
Moreover, a necessary condition for local optimality is that the restricted Hessian satisfies
\begin{align}
\hess|_{\mathcal{C}}(f-\lambda \nabla g)(x^*)\preceq 0\in\mathbb{R}^{(d-1)\times (d-1)}
\label{e73}
\end{align}
where $\lambda$ is as in \eqref{e72} and $\mathcal{C}=\{x\in \mathbb{R}^d\colon (x-x^*)^{\top}\nabla g(x^*)=0\}$ is a codimension one subspace.
Indeed, if \eqref{e73} is not true, then there exists a smooth curve  $(-1,1)\to\mathcal{C}$, $t\mapsto \bar{x}(t)$ satisying $\bar{x}(0)=x^*$ and 
\begin{align}
(f-\lambda g)(\bar{x}(t))
&\ge
(f-\lambda g)(x^*)
+at^2+o(t^2)
\\
&\ge f(x^*)+at^2+o(t^2)
\end{align}
for some $a>0$.
Assuming $\nabla g(x^*)\neq 0$, we can find a smooth curve $(-1,1)\to \{x\colon g(x)=0\}$, $t\mapsto x(t)$ satisfying  $\|x(t)-\bar{x}(t)\|=O(t^2)$.
Then since $\nabla(f-\lambda g)(x^*)=0$ and 
$\|\bar{x}(t)-x^*\|=O(t)$,
we have $\|\nabla(f-\lambda g)(x(t))\|=O(t)$ and
\begin{align}
f(x(t))
&=(f-\lambda g)(x(t))
\\
&\ge 
(f-\lambda g)(\bar{x}(t))+O(t^3)
\\
&\ge f(x^*)+at^2+o(t^2)
\label{e78}
\end{align}
implying that $x^*$ is not a local maximum of the constrained optimization problem.

Let us review the basics about the formal Riemannian structure associated with the optimal transport distance; more background on this topic can be found in, e.g.\ \cite{otto2000generalization}\cite{ambrosio2008gradient}\cite{craig2016exponential}.
Let $\mathcal{P}_2(\mathbb{R}^d)$ be the set of absolutely continuous probability measures with finite second moments.
The Wasserstein-2 distance induces a local inner product structure, which can be viewed as generalization of  the standard (finite dimensional) metric tensor in the classical  Riemannian geometry.
More precisely, at any point $P\in\mathcal{P}_2(\mathbb{R}^d)$, the tangent space is given by 
\begin{align}
\overline{\{
\nabla\phi\colon\phi\in C_c^{\infty}(\mathbb{R}^d)
\}}^{L^2(P)},    \label{e70}
\end{align}
i.e., the $L^2(P)$-closure of the gradients of smooth, compactly supported functions.\footnote{Here we adopt the notations of \cite{clerc2020variational}, where a tangent vector is identified with the velocity field. 
It is worth mentioning that some authors \cite{villani2009optimal} instead identified a tangent vector with the rate of change in the density which is $-{\rm div}(P\nabla \phi)$ by the continuity equation.}
Taking the closure ensures that the tangent space is complete with respect to the inner product defined below.
Given a tangent vector $\nabla \phi$, 
a constant speed geodesic is the curve $t\mapsto P_t$ where $P_t$ is the push-forward of the map $x\mapsto x+t\nabla\phi(x)$.
Note that $\nabla\phi$ as a vector field on $\mathbb{R}^d$ is curl-free, which ensures that such a map is an optimal transport for small enough $t$.
We define the following inner product on the tangent space at $P$:
\begin{align}
\left<\nabla\phi,\nabla\psi\right>
:=\int_{\mathbb{R}^d}\nabla\phi(x)\cdot\nabla\psi(x) dP(x),
\end{align}
which is consistent with the Wasserstein-2 distance since $\|\nabla\phi\|:=\sqrt{\left<\nabla\phi,\nabla\phi\right>}=\frac{d}{dt}W_2(P_t,P)|_{t=0}$.
Furthermore, given a function $f$ on $\mathcal{P}_2(\mathbb{R}^d)$, its gradient and Hessian at $P$ (if exist) satisfy
\begin{align}
f(P_t)=f(P)+t\left<\grad f(P),
\nabla \phi
\right>
+\frac{t^2}{2}\hess f(P)(\nabla \phi,\nabla\phi)+o(t^2)
\label{e81}
\end{align}
for any $\nabla\phi$,
where $\hess f(P)$ is a bilinear form which sends two tangent vectors to a real number.

Under a change of the inner product on the tangent space at $P$, the gradient at $P$ changes linearly.
For example, if $\frac{\delta f}{\delta P}$ denotes the gradient with respect to $L^2(\mathbb{R}^d)$, then the $W_2$-gradient is given by $\grad f(P)=\nabla\frac{\delta f}{\delta P}$.
On the other hand, the Hessian depends on the second order behavior of the geodesic and therefore depends not only on the metric tensor at $P$ but also on the connection (how a tangent vector is parallel transported in a neighborhood).
As such, with a change of metric, the transformation of the Hessian depends on the connection and the gradient and therefore semidefiniteness is not preserved.
However, at a stationary point, semidefiniteness does not depend on the choice of metric.

In the rest of the section we shall focus on the setting of $d=1$, which is enough for constructing a counterexample to Conjecture~\ref{conj1}. 
On $\mathbb{R}$, the vector field $\nabla \phi$ can be though of as a smooth, compactly supported function, which we denote by $U(\cdot)$ or $V(\cdot)$.
The constraint $\int U(x)dx=\int \nabla\phi(x) dx=0$ can be dropped since we will take the closure in $L^2(P)$.
We shall first derive a lemma about the $W_2$-gradient and Hessian of the differential entropy of convolution of measures $h(p*q*r)$, where $p$ and $q$ are viewed as variables whereas $r$ is a fixed probability measure.
The subscripts of $\grad$ and $\hess$ denote the variables which are differentiated in (a formal definition can be given using a formula similar to \eqref{e81}).


\subsection{Hessian Calculation}
\begin{lemma}
Let $r$ be a fixed probability distribution on $\mathbb{R}$, and consider $h(p*q*r)$ as a functional of a pair of probability distributions $(p,q)$.
Denote by $\mu:=p*q*r$, viewed as a probability density function on $\mathbb{R}$.
Let $Z=X+Y+R$, and let $U=U(X)$ and $V=V(Y)$ be two arbitrary smooth, compactly functions of $X$ and $Y$, respectively.
Then we have
\begin{align}
\grad_p h(p*q*r)
&=-\mathbb{E}[\nabla \ln \mu(Z)|X];
\\
\grad_q h(p*q*r)
&=-\mathbb{E}[\nabla \ln \mu(Z)|Y],
\end{align}
and
\begin{align}
\hess_{pp}h(p*q*r)(U,U)
&=-\mathbb{E}[U^2\Delta \ln\mu(Z)]
-\mathbb{E}
\left[
\left(
\frac{\nabla(\mathbb{E}[U|Z]\mu(Z))}{\mu(Z)}
\right)^2
\right];
\\
\hess_{qq}h(p*q*r)(V,V)
&=-\mathbb{E}[V^2\Delta \ln\mu(Z)]
-\mathbb{E}
\left[
\left(
\frac{\nabla(\mathbb{E}[V|Z]\mu(Z))}{\mu(Z)}
\right)^2
\right];
\\
\hess_{pq}h(p*q*r)(U,V)
&=-\mathbb{E}[UV\Delta \ln\mu(Z)]
\nonumber\\
&\quad
-\mathbb{E}
\left[
\frac{\nabla(\mathbb{E}[U|Z]\mu(Z))\nabla(\mathbb{E}[V|Z]\mu(Z))}{\mu^2(Z)}
\right].
\end{align}
\end{lemma}
\begin{proof}
Let $X$, $Y$, $R$, $U$, $V$ be as in the statement of the lemma.
Define for each $s\ge0$, $t\ge0$,
\begin{align}
Z_{st}:=X+sU+Y+tV+R=Z+sU+tV.
\end{align}
Denote by $\mu_{st}$ the distribution of $Z_{st}$.
We have
\begin{align}
\frac{d}{ds}h(\mu_{st})
&=-\frac{d}{ds}\int\mu_{st}\ln\mu_{st}
\\
&=-\int\frac{d\mu_{st}}{ds}\ln\mu_{st}
\\
&=\int\nabla(\mathbb{E}[U|Z_{st}]\mu_{st})\ln\mu_{st}
\\
&=-\int \mathbb{E}[U|Z_{st}]\nabla(\ln\mu_{st})\mu_{st}
\\
&=-\mathbb{E}[U\nabla\ln\mu_{st}(Z_{st})]
\label{e91}
\end{align}
where \eqref{e91} used 
\begin{align}
\frac{d}{ds}\mu_{st}
=-\nabla(\mathbb{E}[U|Z_{st}]\mu_{st})\label{e_78}
\end{align}
which
can be shown by the following method:\footnote{The functional representation approach is of course well-known in analysis; an exploration of this viewpoint for information theory problems can be found in \cite{liu2018information}\cite{liu2019smoothing}\cite{liu2020capacity}\cite{liu2020dispersion}\cite{liu2020second}.}
consider an arbitrary smooth test function $f$; we have 
\begin{align}
\int f\mu_{st}=\mathbb{E}[f(Z_{st})].
\end{align}
Differentiating on both sides,
\begin{align}
\int f\frac{d}{ds}\mu_{st}
&=\mathbb{E}[f'(Z_{st})U]
\\
&=\mathbb{E}[f'(Z_{st})\mathbb{E}[U|Z_{st}]]
\\
&=-\int f(Z_{st})\nabla(\mathbb{E}[U|Z_{st}]\mu_{st})
\end{align}
where the last step used integration by parts. Since $f$ is arbitrary we have confirmed \eqref{e_78}.

Next, we have 
\begin{align}
\frac{d^2}{ds^2}h(\mu_{st})
&=-\int\Delta(\mathbb{E}[U^2|Z_{st}]\mu_{st})\ln\mu_{st}
-\int[\nabla(\mathbb{E}[U|Z_{st}]\mu_{st})]^2\frac1{\mu_{st}}
\\
&=-\int\mathbb{E}[U^2|Z_{st}]\mu_{st}\Delta\ln\mu_{st}
-\int[\nabla(\mathbb{E}[U|Z_{st}]\mu_{st})]^2\frac1{\mu_{st}},
\label{e84}
\end{align}
where the last step used integration by parts, and the first step used 
\begin{align}
\frac{d}{ds}(\mathbb{E}[U|Z_{st}]\mu_{st})
=-\nabla(\mathbb{E}[U^2|Z_{st}]\mu_{st})\label{e85}
\end{align}
which can be shown as follows:
let $f$ be an arbitrary smooth test function; we have 
\begin{align}
\frac{d}{ds}\int\mathbb{E}[U|Z_{st}=z]
f(z)\mu_{st}(z)dz
&=\frac{d}{ds}\mathbb{E}[Uf(Z_{st})]
\\
&=\mathbb{E}[Uf'(Z_{st})U]
\\
&=\int \mathbb{E}[U^2|Z_{st}=z]f'(z)\mu_{st}(z)dz
\end{align}
which implies \eqref{e85} via integration by parts. Now taking $s,t=0$ in \eqref{e84} establishes the formula for $\hess_{pp}h(p*q*r)(U,U)$ claimed in the lemma, and the proofs for the other Hessian components are similar.
\end{proof}
\begin{remark}
If we take $Y,V,R$ to be constants, we recover the following formula for the Hessian of differential entropy
\begin{align}
\hess h(\mu)(U,U)
&=-\int U^2\mu\Delta \ln\mu
-\int[\nabla(U\mu)]^2\frac1{\mu}
\\
&=-\int U^2\mu\Delta\ln\mu
+\int U\mu\nabla\left(\frac{\nabla(U\mu)}{\mu}\right)
\\
&=-\int U^2\mu\Delta\ln\mu
+\int U\mu\nabla\left(\nabla U+U\nabla \ln\mu\right)
\\
&=\int U\mu (\Delta U+\nabla U\nabla\ln \mu)
\\
&=-\int (\nabla U)^2\mu
\end{align}
which is well-known in literature (e.g. \cite[p13]{otto2000generalization}).
\end{remark}
Using the same method we can compute the $W_2$ gradient and the Hessian of the variance functional (details omitted):
\begin{lemma}
Let $p$ be a probability density function on $\mathbb{R}$ with zero mean, and let $U$ be a smooth, compactly supported function on $\mathbb{R}$.
\begin{align}
\grad\var(p)&=2x;
\\
\hess\var(p)(U)&=2\var(U).
\end{align}
\end{lemma}
\begin{theorem}\label{thm3}
Given any $u>0$ and $\lambda\ge0$, define the functionals
\begin{align}
\Psi(p,q)
&=
uh(p*q*\gamma_u)+h(p)-(1+u)h(p*\gamma_u);
\\
\Psi_{\lambda}(p,q)&=\Psi(p,q)-\lambda\var(q).
\end{align}
Suppose that $K>0$ and $L>0$ are such that $(K,L)$ is a stationary point of the function $(K,L)\mapsto \Psi_{\lambda}(\gamma_K,\gamma_L)$.
Then $\hess \Psi_{\lambda}(\gamma_K,\gamma_L)$ restricted on the subspace \begin{align}
\mathcal{C}:=\{(U,V)\colon\int x_2V(x_2)\gamma_L(x_2)dx_2=0\}
\label{e97}
\end{align}
is negative-semidefinite if and only if $K\le \frac{u}{(1+u)^{1/3}-1}$.
In particular, if $K>\frac{u}{(1+u)^{1/3}-1}$ then $(\gamma_K,\gamma_L)$ is not a local maximum of $\Psi(p,q)$ subject to $\var(q)\le L$.
\end{theorem}
Similar to Corollary~\ref{cor1}, Theorem~\ref{thm3} provides a counterexample to Conjecture~\ref{conj1} for the case where $N_1=\Sigma_1=0$ and $N_2=u$ (extension to the $N_1,\Sigma_1>0$ case is possible using a similar continuity argument as Theorem~\ref{thm_continuity}).
\begin{proof}[Proof of Theorem~\ref{thm3}]
Let us assume the expansions
\begin{align}
U&=\sum_{\alpha\in\{0,1,2,\dots\}}A_{\alpha}\frac{D^{\alpha}\gamma_K}{\gamma_K},
\\
V&=\sum_{\alpha\in\{0,1,2,\dots\}}B_{\alpha}\frac{D^{\alpha}\gamma_L}{\gamma_L}.
\end{align}
We can now explicitly compute $\hess\Psi_{\lambda}(\gamma_K,\gamma_L)$ using the following facts about Hermite polynomials (see \cite{cramer1999mathematical}):
\begin{align}
\int\left(\frac{D^{\alpha}\gamma_K}{\gamma_K}\right)^2\gamma_K
=\frac{\alpha}{K^{\alpha}},\quad\forall \alpha=0,1,2,\dots
\end{align}
and
\begin{align}
\mathbb{E}\left[\left.
\frac{D^{\alpha}\gamma_K(X_1)}{\gamma_K(X_1)}\right|
\hat{X}_1
\right]
=
\frac{D^{\alpha}\gamma_{K+u}(\hat{X}_1)}{\gamma_{K+u}(\hat{X}_1)},
\end{align}
where $X_1\sim \gamma_K$,  $\hat{X}_1=X_1+Z_2$, and $Z_2\sim \gamma_u$ is independent of $X_1$.
We have
\begin{align}
\hess_{11}\Psi_{\lambda}(\gamma_K,\gamma_L)(U,U)
&=\frac{u}{K+u+L}\sum_{\alpha\ge 0}A_{\alpha}^2\frac{\alpha!}{K^{\alpha}}
-u\sum_{\alpha\ge 0}A_{\alpha}^2\frac{(\alpha+1)!}{(K+u+L)^{\alpha+1}}
\nonumber\\
&\quad +\frac1{K}\sum_{\alpha\ge 0}A_{\alpha}^2\frac{\alpha!}{K^{\alpha}}
-\sum_{\alpha\ge 0}A_{\alpha}^2\frac{(\alpha+1)!}{K^{\alpha+1}}
\nonumber\\
&\quad-\frac{1+u}{K+u}\sum_{\alpha\ge0}A_{\alpha}^2
\frac{\alpha!}{K^{\alpha}}
+(1+u)\sum_{\alpha\ge 0}A_{\alpha}^2\frac{(\alpha+1)!}{(K+u)^{\alpha+1}}
\end{align}
where the subscript of $\hess_{11}$ means that the function is differentiated in the first argument twice.
Since $K$ is a stationary point, the first, third, and fifth summations cancel out,
and the $\alpha=0$ summands in the other three summations cancel out. We are left with
\begin{align}
\hess_{11}\Psi_{\lambda}(\gamma_K,\gamma_L)(U,U)
&=
-u\sum_{\alpha\ge 1}A_{\alpha}^2\frac{(\alpha+1)!}{(K+u+L)^{\alpha+1}}
\nonumber\\
&\quad 
-\sum_{\alpha\ge 1}A_{\alpha}^2\frac{(\alpha+1)!}{K^{\alpha+1}}
\nonumber\\
&\quad
+(1+u)\sum_{\alpha\ge 1}A_{\alpha}^2\frac{(\alpha+1)!}{(K+u)^{\alpha+1}}.
\end{align}
Next, since $B_1=0$ by \eqref{e97}, we have
\begin{align}
\hess_{22}\Psi_{\lambda}(\gamma_K,\gamma_L)(V,V)
&=\frac{u}{K+u+L}\sum_{\alpha\neq 1}B_{\alpha}^2\frac{\alpha!}{L^{\alpha}}
-u\sum_{\alpha\neq 1}B_{\alpha}^2\frac{(\alpha+1)!}{(K+u+L)^{\alpha+1}}
\nonumber\\
&\quad-2\lambda\sum_{\alpha\neq 1}
B_{\alpha}^2\frac{\alpha!}{L^{\alpha}}+2\lambda B_0^2
\\
&=
\frac{u}{K+u+L}\sum_{\alpha\ge2}B_{\alpha}^2\frac{\alpha!}{L^{\alpha}}
-u\sum_{\alpha\ge2}B_{\alpha}^2\frac{(\alpha+1)!}{(K+u+L)^{\alpha+1}}
\nonumber\\
&\quad-2\lambda\sum_{\alpha\ge2}
B_{\alpha}^2\frac{\alpha!}{L^{\alpha}}
\\
&=-u\sum_{\alpha\ge2}B_{\alpha}^2\frac{(\alpha+1)!}{(K+u+L)^{\alpha+1}}
\end{align}
where the last step follows since $\partial_L\Psi_{\lambda}(\gamma_K,\gamma_L)=0$ implies $\lambda=\frac{u}{K+u+L}$.
Finally
\begin{align}
\hess_{12}\Psi_{\lambda}(\gamma_K,\gamma_L)(U,V)
&=\frac{u}{K+u+L}A_0B_0
-u\sum_{\alpha\neq 1}A_{\alpha}B_{\alpha}\frac{(\alpha+1)!}{(K+u+L)^{\alpha+1}}
\\
&=-u\sum_{\alpha\ge 2}A_{\alpha}B_{\alpha}\frac{(\alpha+1)!}{(K+u+L)^{\alpha+1}}.
\end{align}
Now we can write
\begin{align}
\hess\Psi_{\lambda}(\gamma_K,\gamma_L)(U,V)
&=\hess_{11}\Psi_{\lambda}(\gamma_K,\gamma_L)(U,U)+\hess_{22}\Psi_{\lambda}(\gamma_K,\gamma_L)(V,V)
\nonumber\\
&\quad+2\hess_{12}\Psi_{\lambda}(\gamma_K,\gamma_L)(U,V)
\\
&=\sum_{\alpha\ge1}I_{\alpha}
\end{align}
where we defined
\begin{align}
I_1:=-\frac{2uA_1^2}{(K+u+L)^2}
-\frac{2A_1^2}{K^2}
+\frac{2(1+u)A_1^2}{(K+u)^2};
\end{align}
and for $\alpha\ge2$,
\begin{align}
\frac{I_{\alpha}}{(\alpha+1)!}
&:=
-\frac{uA_{\alpha}^2}{(K+u+L)^{\alpha+1}}
-\frac{A_{\alpha}^2}{K^{\alpha+1}}
+\frac{(1+u)A_{\alpha}^2}{(K+u)^{\alpha+1}}
\nonumber\\
&\quad-\frac{uB_{\alpha}^2}{(K+u+L)^{\alpha+1}}-\frac{2uA_{\alpha}B_{\alpha}}{(K+u+L)^{\alpha+1}}.
\end{align}
Note that $I_1\le 0$, which follows since $\partial_K^2\Psi_{\lambda}(\gamma_K,\gamma_L)<0$ at the maximizer $K$.
Moreover given $A_{\alpha}$ we have 
\begin{align}
\frac{I_{\alpha}}{(\alpha+1)!}
\le -\frac{A_{\alpha}^2}{K^{\alpha+1}}
+\frac{(1+u)A_{\alpha}^2}{(K+u)^{\alpha+1}}
\end{align}
with equality achieved only when $B_{\alpha}=-A_{\alpha}$.
Then the first claim of the theorem follows,
noting that $K=\frac{u}{(1+u)^{1/3}-1}$ is the solution to $-\frac1{K^3}+\frac{1+u}{(K+u)^3}=0$.
The case where $\hess\Psi_{\lambda}$ is not negative-semidefinite can be seen by choosing $B_2=-A_2\neq0$ and $A_{\alpha}=B_{\alpha}=0$ for $\alpha\neq 2$.
Once the restricted Hessian is not negative-semidefinite, 
we can show that the stationary point does not achieve local maximum with arguments similar to \eqref{e78}:
we consider $(p_t,q_t)$ where $p_t$ and $q_t$ are Wasserstein geodesics satisfying $p_0=\gamma_K$, $q_0=\gamma_L$, $\frac{d}{dt}\int x^2q_t|_{t=0}=0$ and $\frac{d^2}{dt^2}\Psi_{\lambda}(p_t,q_t)|_{t=0}>0$.
We can assume without loss of generality that the corresponding tangent vectors \eqref{e70} come from gradients of smooth compactly supported functions.
Then let $\bar{q}_t$ be a rescaling of $q_t$ so that $\var(q_t)=L$, $t>0$.
Then the same argument as in \eqref{e78} shows that $(p_t,\bar{q}_t)$ beats the Gaussian stationary point for small $t>0$,
{\it mutatis mutandis}.
\end{proof}

\section{Application: Gaussian Suboptimality Under Constant Power Control}
\label{sec_application}
The counterexamples in previous sections disproved Conjecture~\ref{conj1},
which is a sufficient but not necessary condition for Gaussian optimality for the Han-Kobayashi (HK) region.
In this section we show that  Gaussian signaling is suboptimal for the HK region among protocols with constant power control.
This is a new result, although it was known that Gaussian signaling with constant power control can be strictly improved by Gaussian signaling with variable power control \cite{costa2011noisebergs}.

We first recall some results in \cite{costa2020structure}.
Given two distributions $P_{X_1},P_{X_2}$ on $\mathbb{R}^d$,
define
\begin{align}
\Psi(P_{X_1},P_{X_2})
&:=
uh(X_1+X_2+Z_1+Z_2)+h(X_1+Z_1)
\nonumber\\
&\quad-(u+1)h(X_1+Z_1+Z_2),
\label{e115}
\end{align}
where $Z_j\sim \mathcal{N}(0,N_jI)$, $j=1,2$, $X_1,X_2,Z_1,Z_2$ are independent, and $u\ge 1$.
Given $P_{X_1}$ and $P_{X_2}$, define the concave envelope by 
\begin{align}
 \mathcal{C}_{X_1}[\Psi(P_{X_1},P_{X_2})]:=&
\sup_{P_{X_1U}}\left\{h(X_1+X_2+Z_1+Z_2|U)
\right.
\nonumber\\
+h(X_1+Z_1|U)&\left.-(u+1)h(X_1+Z_1+Z_2|U)
\right\}
\label{e35}
\end{align}
where the supremum is over $P_{X_1U}$ compatible with the given $P_{X_1}$, and $(U,X_1)$ is independent of $(X_2,Z_1,Z_2)$.
Define
\begin{align}
F_{u}(P_{X_1},P_{X_2})&:=
h(X_1+X_2+Z_1+Z_2)-h(Z_1)
\nonumber\\
&+\mathcal{C}_{X_1}[\Psi(P_{X_1},P_{X_2})].
\end{align}

Recall that the $d$-letter Han-Kobayashi (HK) region for the Gaussian Z-channel is given by (see e.g.\ \cite{costa2017gaussian})
\begin{align}
dR_1&\le h(X_1+Z_1|Q)-h(Z_1)
\label{e_37}
\\
dR_2&\le h(X_1+X_2+Z_1+Z_2|U_1,Q)
\nonumber\\
&-h(X_1+Z_1+Z_2|U_1,Q)
\\
d(R_1+R_2)&\le h(X_1+X_2+Z_1+Z_2|Q)
\nonumber\\
&-h(X_1+Z_1+Z_2|U_1,Q)
\nonumber\\
&+h(X_1+Z_1|U_1,Q)-h(Z_1)
\label{e39}
\end{align}
where $Z_1\sim \mathcal{N}(0,I)$, $Z_2\sim \mathcal{N}(0,N_2I)$. $Q$ is the power control random variable, and the joint distribution factors as
$P_QP_{U_1X_1|Q}P_{X_2|Q}P_{Z_1}P_{Z_2}$.
Following \cite{costa2020structure}, we note that for $u\ge 0$, the HK bound can be expressed using the weighted form
\begin{align}
&\quad d(R_1+(1+u)R_2)
\nonumber\\
&\le \sup_{P_{Q_1Q_2}\colon \mathbb{E}[Q_j]\le q_j}\mathbb{E}_{Q_1,Q_2}\left[\sup_{P_{X_j}\colon \mathbb{E}[\|X_j\|^2]\le Q_j}F_{u}(P_{X_1},P_{X_2})\right]
\label{e120}
\end{align}
where the first supremum is over distributions of $(Q_1,Q_2)$ which are random non-negative powers satisfying a given power constraint $\mathbb{E}[Q_j]\le q_j$, $j=1,2$.
For any given $Q_j$, $j=1,2$, the second supremum is over $(P_{X_1},P_{X_2})$ satisfying the indicated power constraints.

We remark that for $u\in(-1,0]$, it is obvious that Gaussian variables are optimal for the weighted sum rates.
Indeed, \eqref{e_37} is optimized by Gaussian $X_1$ under a power constraint.
Gaussian variables are also optimal for \eqref{e39} since $-h(X_1+Z_1+Z_2|U_1,Q)+h(X_1+Z_1|U_1,Q)=-I(X_1+Z_1+Z_2;Z_2|U_1,Q)$
which is maximized by Gaussian variables under a power constraint (by Gaussian saddle point; see e.g. \cite{cover1999elements}).

We now show that under constant power control (i.e.\ $Q_1$ and $Q_2$ are constants) and $u=1$, Gaussian variables may be suboptimal for the weighted sum rate.
\begin{theorem}
There exist $N_1=1$, $N_2>0$, $u=1$, $d=1$, and  some (deterministic) $q_1,q_2>0$ such that Gaussian $P_{X_j}$, $j=1,2$ are not optimal for the following optimization:
\begin{align}
\sup_{P_{X_j}\colon \mathbb{E}[\|X_j\|^2]\le q_j}F_{u}(P_{X_1},P_{X_2}).  
\end{align}
\end{theorem}
\begin{proof}
Using the argument around \eqref{e25} we see that there exists $N_1=1$, $u=1$ and $N_2>0$ such that 
\begin{align}
\sup_{P_{X_1'},P_{X_2}\colon\mathbb{E}[X_2^2]\le N_2}\Psi(P_{X_1'},P_{X_2})
>0.
\end{align}
Now we consider any $X_1'$ and $X_2$ such that $\Psi(P_{X_1'},P_{X_2})=c>0$ and $\mathbb{E}[X_2^2]\le N_2$.
Let $U\sim \mathcal{N}(0,A)$ independent of $(X_1',X_2,Z_1,Z_2)$ where $A>0$ will be chosen later.
Let $X_1=X_1'+U$.
Then we have
\begin{align}
\mathcal{C}_{X_1}[\Psi(P_{X_1},P_{X_2})]\ge \Psi(P_{X_1'},P_{X_2})=c.
\end{align}
Moreover, by choosing $A$ large enough (while other parameters are kept fixed), we have 
\begin{align}
&h(X_1+X_2+Z_1+Z_2)-h(Z_1)
\\
&\ge\frac1{2}\ln\frac{\var(X_1')+A+N_1+2N_2}{N_1}
-\frac{c}{2}.
\end{align}
For such $A$ and with $q_1:=\var(X_1')+A$, $q_2:=N_2$ we have
\begin{align}
&\sup_{P_{X_j}\colon \mathbb{E}[X_j^2]\le q_j,j=1,2} F_{u}(P_{X_1},P_{X_2})
\nonumber\\
&\ge \frac1{2}\ln\frac{\var(X_1')+A+N_1+2N_2}{N_1}
+\frac{c}{2}.
\label{e_41}
\end{align}
On the other hand, if the supremum in \eqref{e_41} is restricted to Gaussian $X_1,X_2$ with the same variances, then it was shown in \cite{costa2020structure} that Gaussian $(U,X_1)$ is optimal in \eqref{e35}.
By the result of Lemma~\ref{lem2} we have that $\mathcal{C}_{X_1}[\Psi(P_{X_1},P_{X_2})]\le 0$.
Moreover under $\mathbb{E}[X_j^2]\le q_j,j=1,2$ it is also obvious that \begin{align}
&h(X_1+X_2+Z_1+Z_2)-h(Z_1)
\nonumber\\
&\le \frac1{2}\ln\frac{\var(X_1')+A+N_1+2N_2}{N_1}.
\end{align}
The same upper bounds holds for the left side of \eqref{e_41} if $X_1,X_2$ are restricted to be Gaussian, and the claim of the theorem follows by comparing it with \eqref{e_41}.
\end{proof}

\section{Structure of the Gaussian Extremizers}
\label{sec_structure}
Although counterexamples in the previous sections disproved Conjecture~\ref{conj1},
they do not show Gaussian suboptimality for the Han-Kobayashi (HK) region allowing power control, as we explain in this section.
The reason is that in Corollary~\ref{cor1} the counterexample exists only when the Gaussian maximizer $P_{X_1}$ has covariance larger than
$\frac{u}{(1+u)^{1/3}-1}$,
but this section will show that the Gaussian maximizer for the HK region (allowing power control) never falls into this case. 
Let us remark that Conjecture~\ref{conj1} is a necessary but not sufficient condition for Gaussian optimality of the HK region allowing power control, 
and it is not clear if  Conjecture~\ref{conj1} can be easily amended (e.g.\ restricting to the case where the covariance of the Gaussian optimizer satisfies a certain bound) and remains a necessary condition.

We define the following quantities that characterize the HK region \eqref{e120} restricted to Gaussian signaling:
Given $u,N_1>0$ and positive semidefinite matrices $K$ and $L$ of the same dimensions, define
\begin{align}
\psi(K,L)&:=u\ln\det(K+N_1I+uI+L)+\ln \det (K+N_1I)
\nonumber\\
&\quad-(u+1)\ln\det(K+N_1I+uI).
\label{e129}
\end{align}
For positive semidefinite matrices $K$ and $L$ of the same dimensions, define
\begin{align}
\phi(J,L)&:=\sup_{K\preceq J}\psi(K,L)
\end{align}
where $K$ is positive semidefinite and of the same dimensions as $J$.
For $q_1,q_2\ge0$ and a positive integer $d$, define
\begin{align}
f_d(q_1,q_2)&:=\sup_{\tr(J)\le q_1,\,\tr(L)\le q_2}\{\ln\det(J+N_1I+uI+L)+\phi(J,L)\},
\end{align}
where the supremum is over positive semidefinite $d\times d$ matrices $J$ and $L$ satisfying the indicated constraints.
For $q_1,q_2\ge0$ define
\begin{align}
g_d(q_1,q_2)
=\sup_{\mathbb{E}[Q_1]\le q_1,\,\mathbb{E}[Q_2]\le q_2}\mathbb{E}[f_d(Q_1,Q_2)]
\end{align}
where the supremum is over nonnegative (possibly dependent) random variables  $Q_1$ and $Q_2$ satisfying the indicated bounds.
Obviously $f_d(q_1,q_2)\le g_d(q_1,q_2)$, and the inequality may be strict for some $(q_1,q_2)$.

\subsection{One Dimensional Case}
We first consider the case where the dimension $d=1$.
\begin{lemma}\label{lem5}
Suppose that $J,K,L>0$,
and $f_1(J,L)=g_1(J,L)$.
Assume that $K$ achieves the supremum in the definition of $\phi(J,L)$.
Then $K+N_1\le 1+\sqrt{1+u}$.
\end{lemma}
\begin{proof}
We will assume, without loss of generality, that $N_1=0$, since otherwise either $K=0$ or we can make the substitution $K\leftarrow K+N_1$ and the rest of the proof will carry through.

From the definition of $\psi$ we can verify that given $L$ and $N_1$,
\begin{align}
\argmax_{K\ge 0}\psi(K,L)
=\left\{
\begin{array}{cc}
     \infty & L\le 1;  \\
     \frac{u+L}{L-1}&
     L>1.
\end{array}
\right.
\end{align}
Therefore $K:=\argmax_{0\le K\le J}\psi(K,L)=J\wedge \frac{u+L}{L-1}$ if $L>1$ and $K=J$ otherwise.
Now consider $J_t:=J+t$, $L_t:=L-t$, where $t\in\mathbb{R}$.
\\
Case~1: $L>1$ and $J>\frac{u+L}{L-1}$. 
In this case $L_t>1$ and $J_t>\frac{u+L_t}{L_t-1}$ for small enough $t$, so that
\begin{align}
0&\ge \partial_t^2f_1(J_t,L_t)|_{t=0}
\\
&=
\partial_t^2\phi(J_t,L_t)|_{t=0}
\\
&=\partial_t^2\psi(\tfrac{u+L_t}{L_t-1},L_t)|_{t=0}
\end{align}
where the first inequality is implied by $f_1(J,L)=g_1(J,L)$.
Since
\begin{align}
\psi\left(\frac{u+L}{L-1},L\right)=
(u+1)\ln(u+L)
-\ln L
-(u+1)\ln(u+1),
\end{align}
by taking the second derivative we see that $L\ge \sqrt{u+1}+1$ and hence $K=\frac{u+L}{L-1}\le \sqrt{u+1}+1$.
\\
Case~2: either $L\le 1$, or $L>1$ and $J<\frac{u+L}{L-1}$.
In this case, for small enough $t$,
\begin{align}
0&\ge \partial_t^2f(J_t,L_t)|_{t=0}
\\
&=
\partial_t^2\phi(J_t,L_t)|_{t=0}
\\
&=\partial_t^2\psi(J_t,L_t)|_{t=0}
\label{e119}
\end{align}
where the first inequality is implied by $f_1(J,L)=g_1(J,L)$.
By computing the second derivative in \eqref{e119} we see that $-\frac1{J^2}+\frac{1+u}{(J+u)^2}\le 0$ and hence $K=J\le 1+\sqrt{u+1}$.
\\
Case~3: $L>1$ and $J=\frac{u+L}{L-1}$.
In this case, let us assume (for proof by contradiction) that $K>1+\sqrt{u+1}$. Then $K=J=\frac{u+L}{L-1}$ implies that  $L<1+\sqrt{u+1}$. Since $\left.\frac{d}{dl}\left(\frac{u+l}{l-1}\right)\right|_{l=L}
<\left.\frac{d}{dl}\left(\frac{u+l}{l-1}\right)\right|_{l=1+\sqrt{u+1}}=-1$, we see that 
$J_t\le \frac{u+L_t}{L_t-1}$ for $t\ge 0$ and $J_t> \frac{u+L_t}{L_t-1}$ for $t<0$ (for $t$ in some neighborhood of 0). We argue that $\partial_tf_1(J_t,L_t)$ exists and is continuous at $t=0$. Indeed, 
\begin{align}
\partial_tf_1(J_t,L_t)|_{t=0^+}
&=
\partial_t\phi(J_t,L_t)|_{t=0^+}
\\
&=\partial_t\psi(J_t,L_t)|_{t=0}
\\
&=\partial_1\psi(K,L)\partial_t J_t|_{t=0}
+\partial_2\psi(K,L)\partial_t L_t|_{t=0}
\\
&=\partial_2\psi(K,L)\partial_t L_t|_{t=0}
\end{align}
where the last step used the fact that $K=\argmax\psi(\cdot,L)$.
Moreover,
\begin{align}
\partial_tf_1(J_t,L_t)|_{t=0^-}
&=
\partial_t\phi(J_t,L_t)|_{t=0^-}
\\
&=\partial_t\left.\psi\left(\frac{L_t+u}{L_t-1},L_t\right)\right|_{t=0}
\\
&=\partial_1\psi(K,L)\partial_t\left. \left(\frac{L_t+u}{L_t-1}\right)\right|_{t=0}
+\partial_2\psi(K,L)\partial_t L_t|_{t=0}
\\
&=\partial_2\psi(K,L)\partial_t L_t|_{t=0}
\end{align}
Therefore 
$\partial_tf_1(J_t,L_t)|_{t=0^+}
=\partial_tf_1(J_t,L_t)|_{t=0^-}$.
The assumption $f_1(J,L)=g_1(J,L)$ implies that $\partial_t^2f_1(J_t,L_t)|_{t=0^+}\le 0$, and by a similar calculation as Case~2, we have $K\le 1+\sqrt{u+1}$, a contradiction. Thus we must have $K\le 1+\sqrt{u+1}$.
\end{proof}

\subsection{Vector Case}
The goal of this subsection is to extend Lemma~\ref{lem5} to the vector case and show that \emph{any} optimal covariance matrix $K$, if nonzero, must have top eigenvalue upper bounded by $1+\sqrt{1+u}-N_1$.
This relies on the Gaussian tensorization property \cite[Theorem 2]{nair2018invariance}, which says that the Gaussian HK region is exhausted by random vectors with diagonal covariance matrices.
However, since Gaussian optimizers are not unique in the vector case (at least rotations preserves optimality), the conclusion of \cite{nair2018invariance} does not exclude the possibility of \emph{some} Gaussian optimizer whose covariance matrices $(J,K,L)$ are not diagonal in a certain common basis (in fact, this is possible).
In order to bound the top eigenvalue nevertheless, 
we shall prove a slightly stronger statement of Proposition~\ref{prop6} below.
The proof ingredients, in particular the inequality parts in Proposition~\ref{prop4}-\ref{prop5} below, have appeared in \cite{nair2018invariance}.
These linear algebra facts are by no means new, but we provide the short proofs here since they play an essential role towards our goal of this section.
The only if part of Proposition~\ref{prop4} was not mentioned in \cite{nair2018invariance}, but follows immediately from the argument we present here.
\begin{definition}\label{defn1}
Let $K$ be a positive semidefinite matrix.
We say a diagonal matrix $\bar{K}$ is a decreasing (resp.\ increasing) alignment of $K$ if $\bar{K}=Q^{\top}KQ$ for some orthogonal $Q$ and the diagonal entries of $\bar{K}$ are decreasing (resp.\ increasing) from top left to bottom right. 
\end{definition}

\begin{proposition}\label{prop3}
Suppose that $K$ and $L$ are positive semidefinite matrices of the same dimensions.
Consider the function
$Q\mapsto \ln\det(K+Q^{\top}LQ)$ on the set of orthogonal matrices;
$Q=I$ is a stationary point if and only if $K$ and $L$ commute.
\end{proposition}
\begin{proof}
Recall that the tangent space of the manifold of orthogonal matrices at $I$ is the set of anti-symmetric matrices. 
Therefore $Q=I$ is a stationary point if and only if for any anti-symmetric $H$,
\begin{align}
\lim_{t\to0}\frac1{t}\ln\det(K+(I+tH)L(I-tH))&=\lim_{t\to0}\frac1{t}\ln\det(K+L+t(HL-LH))
\\
&=\tr((K+L)^{-1}(HL-LH))
\\
&=0.
\end{align}
The last equality is equivalent to 
$\tr((L(K+L)^{-1}-(K+L)^{-1}L)H)=0$.
Since $L(K+L)^{-1}-(K+L)^{-1}L$ is anti-symmetric and since $H$ is arbitrary, this is equivalent to $L$ commuting with $(K+L)^{-1}$ and hence with $K$.
\end{proof}
\begin{proposition}\label{prop4}
Suppose that $K$ and $L$ are positive semidefinite matrices of the same dimensions, and let $\bar{K}$ and $\bar{L}$ be their decreasing and increasing alignments, respectively.
Then $\ln\det(K+L)\le \ln\det (\bar{K}+\bar{L})$, equality holding only if $\bar{K}=Q^{\top}KQ$ and $\bar{L}=Q^{\top}LQ$ for some orthogonal $Q$.
\end{proposition}
\begin{proof}
We first observe the following: 
given vectors $(\alpha_i)_{i=1}^d$ and $(\beta_i)_{i=1}^d$, let $(\bar{\alpha}_i)_{i=1}^d$ and $(\bar{\beta}_i)_{i=1}^d$ be their decreasing and increasing alignments, respectively.
Then
\begin{align}
\prod_{i=1}^d(\alpha_i+\beta_i)
\le 
\prod_{i=1}^d(\bar{\alpha}_i+\bar{\beta}_i),\label{e136}
\end{align}
which is easy to see by considering how the left side of \eqref{e136} changes when exchanging any two coordinates of $(\beta_i)_{i=1}^d$.
Moreover if \eqref{e136} achieves equality then by induction we can show that there exists a permutation $\pi$ such that 
\begin{align}
\bar{\alpha}_i=\alpha_{\pi(i)};
\quad
\bar{\beta}_i=\beta_{\pi(i)}.
\label{e_eq}
\end{align}
Now
\begin{align}
\ln\det(K+L)
&\le \sup_{Q~{\rm orthogonal}}\ln\det(K+Q^{\top}LQ)
\label{e137}
\\
&\le \sup_{Q\textrm{ orthogonal, }K\textrm{ commutes with }Q^{\top}LQ}
\ln\det(K+Q^{\top}LQ)
\label{e_137}
\\
&\le \ln \det(\bar{K}+\bar{L})
\label{e138}
\end{align} 
where \eqref{e_137} follows by Proposition~\ref{prop3}, and 
\eqref{e138} follows from \eqref{e136}.
Moreover if $\ln \det(K+L)=\ln \det(\bar{K}+\bar{L})$
then \eqref{e137} achieves equality, and so $K$ and $L$ commute by Proposition~\ref{prop3}; 
then the observation \eqref{e_eq} implies that $K$ and $L$ can be diagonalized simultaneously into $\bar{K}$ and $\bar{L}$.
\end{proof}

\begin{proposition}\label{prop5}
Suppose that $K$ and $K$ are positive semidefinite matrices of the same dimensions and satisfying $K\preceq K$.
Let $\bar{K}$ and $\bar{K}'$ be their decreasing alignments. Then $\bar{K}'\preceq\bar{K}$.
\end{proposition}
\begin{proof}
The claim is immediate from the fact that the $i$-th (in the decreasing order) eigenvalue of $K$ equals $\min_{\mathcal{L}}\max_{x}x^{\top}Kx$, where the min is over all subspaces $\mathcal{L}$ with codimension $i-1$, and the max is over all unit vectors $x\in\mathcal{L}$.
\end{proof}

\begin{proposition}\label{prop6}
Suppose that $(J,L)$ achieves the supremum in the definition of $f_d(q_1,q_2)$ (for some $q_1,q_2>0$), and
$K$ achieves the supremum in the definition of $\phi(J,L)$.
Let $\bar{K}$ and $\bar{J}$ be the decreasing alignments of $K$ and $J$, and $\bar{L}$ be the increasing alignment of $L$.
Then $(\bar{J},\bar{L})$ achieves the supremum in the definition of $f_d(q_1,q_2)$, and $\bar{K}$ achieves the supremum in the definition of $\phi(\bar{J},\bar{L})$.
\end{proposition}
\begin{proof}
From Proposition~\ref{prop5} and Proposition~\ref{prop4} we see that
\begin{align}
\phi(\bar{J},\bar{L})\ge\psi(\bar{K},\bar{L})\ge \phi(J,L). 
\label{e141}
\end{align}
Using Proposition~\ref{prop4} again we see that $\ln\det(\bar{J}+N_1I+uI+\bar{L})
+\phi(\bar{J},\bar{L})
\ge 
\ln\det(J+N_1I+uI+L)+\phi(J,L)$ and therefore $(\bar{J},\bar{L})$ must achieve the supremum in the definition of $f_d(q_1,q_2)$.
Next we reverse the signs of inequalities in \eqref{e141}:
Since $(J,L)$ achieves the supremum in the definition of $f_d(q_1,q_2)$, we have
\begin{align}
\ln(J+N_1I+uI+L)+\phi(J,L)
&=f_d(q_1,q_2)
\\
&\ge \ln(\bar{J}+N_1I+uI+\bar{L})
+\phi(\bar{J},\bar{L})
\end{align}
which combined with $\ln(J+N_1I+uI+L)\le \ln(\bar{J}+N_1I+uI+\bar{L})$ shows that $\phi(J,L)\ge\phi(\bar{J},\bar{L})$.
This implies that equalities are achieved in \eqref{e141}, hence $\bar{K}$ achieves the supremum in the definition of $\phi(\bar{J},\bar{L})$.
\end{proof}
One consequence of Proposition~\ref{prop6} is the following tensorization property:
\begin{corollary}
Given any positive integer $d$ and $q_1,q_2\in(0,\infty)$, we have $g_d(dq_1,dq_2)=dg_1(q_1,q_2)$,
\end{corollary}

\begin{remark}
$f_d(dq_1,dq_2)=df_1(q_1,q_2)$ is not true in general.
\end{remark}

\begin{remark}
Proposition~\ref{prop6} goes slightly further than the literature \cite{nair2018invariance}: we not only show that the value of the supremum tensorizes, but also that \emph{all} the maximizers (the matrices achieving the supremum) must tensorize, which will be helpful in our Theorem~\ref{thm4}.
\end{remark}

\begin{theorem}\label{thm4}
Suppose that $(q_1,q_2)$ satisfies $f_d(q_1,q_2)=g_d(q_1,q_2)$,
$(J,L)$ achieves the supremum in the definition of $f_d(q_1,q_2)$, and $K$ achieves the supremum in the definition of $\phi(J,L)$.  If $K$ is nonzero then its eigenvalues are all upper bounded by $1+\sqrt{1+u}-N_1$.
\end{theorem}
\begin{proof}
Let $\bar{J}$ and $\bar{K}$ be the decreasing alignments of $J$ and $K$, respectively, and let $\bar{L}$ be the increasing alignment of $L$.
We have
\begin{align}
\sum_{i=1}^dg_1(\bar{J}_{ii},\bar{L}_{ii})
&\le g_d(q_1,q_2)
\label{e144}
\\
&=f_d(q_1,q_2)
\label{e145}
\\
&\le\sum_{i=1}^df_1(\bar{J}_{ii},\bar{L}_{ii})
\label{e94}
\\
&\le \sum_{i=1}^dg_1(\bar{J}_{ii},\bar{L}_{ii})
\end{align}
where \eqref{e144} can be seen from the definition of $g_d$;
\eqref{e145} is by the assumption of the theorem;
\eqref{e94} follows since
Proposition~\ref{prop6}
shows that $(\bar{J},\bar{L})$ achieves the supremum in the definition of $f_d(q_1,q_2)$, and $\bar{K}$ achieves the supremum in the definition of $\phi(\bar{J},\bar{L})$.
Then equalities are therefore achieved throughout, and since $f_1\le g_1$ we must have $f_1(\bar{J}_{11},\bar{L}_{11})=g_1(\bar{J}_{11},\bar{L}_{11})$.
Since $\bar{K}$ achieves the supremum in the definition of $\phi(\bar{J},\bar{L})$,
$\bar{K}_{11}$ must achieve the supremum in the definition of $\phi(\bar{J}_{11},\bar{L}_{11})$.
Then by Lemma~\ref{lem5} we have $\bar{K}_{11}+N_1\le 1+\sqrt{1+u}$ unless $\bar{K}_{11}=0$.
\end{proof}

\section{Local Gaussian Optimality}\label{sec_local}
Results in the previous sections show that our counterexamples are not sufficient for establishing Gaussian suboptimality for the Han-Kobayashi (HK) region allowing power control.
Note that the counterexample in Corollary~\ref{cor1} relies on perturbation along a geodesic line in $L^2(\mathbb{R})$ and uses the order-3 Hermite polynomial; 
one might wonder if the result can be improved by perturbing along a more sophisticated ``curve'' and using lower order Hermite polynomials.
This is not possible, as the Hessian calculations already suggest. 
However, since the Hessian is not necessarily ``continuous'' with respect to the same metric it is calculated with, 
semidefiniteness of the Hessian at a single point does not provide a rigorous proof of 
local optimality.

In this section,
we show that in the (interior of the) regime where the Hessian at the Gaussian stationary point is negative definite, 
local maximum is indeed achieved.
The proof is based on showing that the Hessian is negative semidefinite in an $L^{\infty}(\gamma_K)$ neighborhood.

The $L^{\infty}(\gamma_K)$ metric appears to be natural for this setting; the celebrated
Holley-Stroock perturbation principle (see e.g.\ \cite[p1184]{holley}\cite{ledoux2001logarithmic}\cite{otto2000generalization})
provides a method of controlling the best constants in functional inequalities such as the Poincare inequality (also concerning a bound on the spectrum of a self-adjoint operator),
under perturbation in the $L^{\infty}$-norm.
This method is simple yet avoids assumption of bounds on higher derivatives.
Our argument may fail if other metrics are used (see Remark~\ref{rem9}).

Recall the functional $\Psi$ defined in \eqref{e115}, which played a role in the expression of the Han-Kobayashi region.
With an abuse of notation, in this section we shall define the following functional:
\begin{align}
\Psi(P_X,L)&:=\sup_{P_Y\colon \cov(P_Y)\preceq L}uh(P_X*\gamma_{uI}*P_Y)
+h(P_X)-(1+u)h(P_X*\gamma_{uI}),
\end{align}
where $u>0$ and positive integer $d$ are given,
$P_X$ is a distribution on $\mathbb{R}^d$, and $L$ is a $d\times d$ positive semidefinite matrix.
Here we have taken $N_1=0$, 
which is without loss of generality for the purpose proving local optimality, since the general case easily follows by taking $X_1+Z_1$ as the new $X_1$.
Also, we have taken the covariance of $Z_2$ to be $uI$, which is without loss of generality by a scaling argument.

We define the following as a simpler proxy for $\Psi(P_X,L)$:
\begin{align}
\Theta(P_X,L)&:=\frac{u}{2}\ln [(2\pi e)^d\det(\cov(P_X)+uI+L)]
+h(P_X)-(1+u)h(P_X*\gamma_{uI}).
\label{e155}
\end{align}
Note that $\Psi(P_X,L)\le\Theta(P_X,L)$, with equality achieved only if $P_X$ is Gaussian (the only part follows from the Levy-Cramer theorem \cite{cramer1936eigenschaft}).
Define
\begin{align}
\Phi(L)&:=
\sup_K\Theta(\gamma_K,L)
\label{e_145}
\end{align}
where the supremum is over all positive semidefinite matrices $K$ with the same dimensions as $L$.
Observe that
the supremum in \eqref{e_145} can be achieved if and only if the least eigenvalue of $L$ is strictly larger than 1;
in that case the maximizer is 
\begin{align}
K=(L-I)^{-1}(L+uI).
\label{e146}
\end{align}
Moreover, the Hessian at the maximizer $K$ is strictly negative definite:
Indeed, by rotation invariance of the log determinant function, 
we can assume without loss of generality that $L$ agrees with its increasing alignment (Definition~\ref{defn1}), in which case $K$ agrees with its decreasing alignment.
For $\Delta$ of the same dimensions as $K$ and whose Frobenius norm $\|\Delta\|$ is sufficiently small, we have
\begin{align}
\Theta(\gamma_{K+\Delta},L)
-\Theta(\gamma_K,L)
&\le
\Theta(\gamma_{K+\tilde{\Delta}},L)
-\Theta(\gamma_K,L)
\label{e147}
\\
&\le -\sum_{i=1}^da_i\tilde{\Delta}_{ii}^2+o(\|\Delta\|^2)
\\
&\le -\min_{1\le i\le d} a_i\|\Delta\|^2
+o(\|\Delta\|^2)
\label{e149}
\end{align}
where $\tilde{\Delta}$ is a diagonal matrix whose diagonal entries are some permutation of the eigenvalues of $\Delta$, 
\eqref{e147} follows from Proposition~\ref{prop3} ahead,
and $a_i>0$ is a function of $L_{ii}$. 
Thus strict negative definiteness of the Hessian is proved.

The main result of this section is the following:
\begin{theorem}\label{thm5}
Given $u>0$ and a positive semidefinite matrix $L$, suppose that $K$ is a maximizer in \eqref{e_145}.
Suppose that the top eigenvalue of $K$ is strictly smaller than $\frac{u}{(1+u)^{1/3}-1}$.
Then exist $\epsilon>0$ such that for all $P_X$ satisfying $\|P_X-\gamma_K\|_ {L^{\infty}(\gamma_K)}\le \epsilon$ and $\int xdP_X(x)=0$, we have
$\Theta(P_X,L)\le \Phi(L)$,
and consequently $\Psi(P_X,L)\le \Phi(L)$.
\end{theorem}
The result of Theorem~\ref{thm5} can be extended (with the same $\frac{u}{(1+u)^{1/3}-1}$ bound on the top eigenvalue) to the case where the supremum in \eqref{e_145} is restricted to $K$ satisfying $K\preceq J$, where $J$ is a positive semidefinite matrix that commutes with $L$; we omit the details of the analysis.
Together with Theorem~\ref{thm4},
Theorem~\ref{thm5} shows that if Gaussian distributions $(P_{UX_1},P_{X_2})$ is a Gaussian stationary point for the supremum in \eqref{e120}, 
then the expression to the right of the supremum in \eqref{e35} cannot be improved by local (in the sense described by the Theorem~\ref{thm5}) perturbation.
\begin{remark}\label{rem9}
The Theorem may fail if $L^{\infty}(\gamma_K)$ is replaced by other metrics. 
For example, take $d=1$ and consider the Gaussian mixture $P_X^{\epsilon}=(1-\epsilon)\gamma_K+\epsilon\mathcal{N}(0,\epsilon^{-2}K)$.
Then $\lim_{\epsilon\downarrow0}\|P_X^{\epsilon}-\gamma_K\|_{L^{\infty}(\mathbb{R})}=0$ but
$\lim_{\epsilon\downarrow0}\|P_X^{\epsilon}-\gamma_K\|_{L^{\infty}(\gamma_K)}=\infty$,
where $L^{\infty}(\mathbb{R})$ denotes the $L^{\infty}$-norm with respect to the Lebesgue measure.
Using the chain rule of entropy we can show that $|h(P_X^{\epsilon})-h(P_X)|=O(\epsilon\ln \frac1{\epsilon})$ and $|h(P_X^{\epsilon}*\gamma_u)-h(P_X*\gamma_u)|=O(\epsilon\ln \frac1{\epsilon})$, but $\var(P_X^{\epsilon})=\var(P_X)+\Theta(\epsilon^{-1})$.
Therefore $\Theta(P_X^{\epsilon},L)> \Phi(L)$ for small $\epsilon$.
We remark that a similar counterexample has been previously proposed for the stability of the log-Sobolev inequality \cite{eldan2020stability}.
In fact, \eqref{e155} is closely related to the stability of log-Sobolev inequality, since given $P_X$,
\begin{align}
\lim_{u\to\infty}\frac1{u}[h(P_X)-(1+u)h(P_X*\gamma_{uI})]
&=-\frac1{2}J(P_X)-h(P_X) 
\end{align}
where $J(P_X)$ denotes the Fisher information (see e.g.\ \cite{cover1999elements}).
\end{remark}

\begin{proof}[Proof of Theorem~\ref{thm5}]
Pick an arbitrary $r$ such that $\|r\|_{L^{\infty}(\gamma_K)}\le \epsilon$ ($\epsilon<1$ to be chosen later), $\int r=0$, and $\int xr=0$.
Define $p_t:=\gamma_K+tr$ and $q_t:=p_t*\gamma_{uI}$, where $t\in[0,1]$. 
By assumption $K$ is the maximizer in \eqref{e146} and hence a stationary point,  calculations show that  
$\frac{d}{dt}\Theta(p_t,L)|_{t=0}=0$ (even though $p_t$ is not necessarily Gaussian).
Thus to prove the theorem it remains show the negativity of the second derivative for $t\in[0,1]$.
We have
\begin{align}
\frac{d}{dt}\cov(p_t)
&=\frac{d}{dt}\int xx^{\top}p_t
\\
&=\int xx^{\top}r
\end{align}
and 
$
\frac{d^2}{dt^2}\cov(p_t)
=-2\int xr\cdot\int x^{\top}r
$.
Therefore
\begin{align}
\frac{d}{dt}\ln\det(\cov(p_t)+uI+L)
=\tr\left(
(\cov(p_t)+uI+L)^{-1}\frac{d}{dt}\cov(p_t)
\right)
\end{align}
and
\begin{align}
&\quad\frac{d^2}{dt^2}\ln\det(\cov(p_t)+uI+L)
\nonumber
\\
&=-\tr\left((\cov(p_t)+uI+L)^{-1}
\frac{d}{dt}\cov(p_t)
(\cov(p_t)+uI+L)^{-1}
\frac{d}{dt}\cov(p_t)
\right)
\\
&\le -\frac1{(1+t\epsilon)^2}\tr\left(
(K+uI+L)^{-1}\int xx^{\top}r
(K+uI+L)^{-1}\int xx^{\top}r
\right)
\label{e154}
\end{align}
where we used the fact that 
$\cov(p_t)=\int xx^{\top}p_t\le \int xx^{\top}(1+t\epsilon)\gamma_K\le (1+t\epsilon)K$.
Moreover, from \eqref{e_ent} we obtain
\begin{align}
\frac{d^2}{dt^2}h(p_t)&=-\int\frac{r^2}{p_t}
\\
&\le -\frac1{1+t\epsilon}\int\frac{r^2}{\gamma_K}
\label{e156}
\end{align}
and similarly
\begin{align}
&\quad\frac{d^2}{dt^2}h(q_t)
\nonumber\\
&=-\int\frac{(r*\gamma_{uI})^2}{p_t*\gamma_{uI}}
\\
&\ge
-\frac1{1-t\epsilon}\int\frac{(r*\gamma_{uI})^2}{\gamma_{K+uI}}
\label{e158}
\end{align}
where we used $(1-t\epsilon)\gamma_K(x)\le p_t(x)\le (1+t\epsilon)\gamma_K(x)$,
for all $x\in\mathbb{R}^d$.

In the rest of the proof assume without loss of generality that $K$ is diagonal with $k_1,\dots,k_d$ being the diagonal entries.
Consider the expansion $r=\sum_{\alpha\in\{0,1,\dots\}^d}\lambda_{\alpha}D^{\alpha}\gamma_K$, and plug it into the right sides of \eqref{e154},\eqref{e156},\eqref{e158}.
It is easy to see using \eqref{e_herm} that we have
\begin{align}
\frac{d^2}{dt^2}\Theta(p_t,L)\le \sum_{\alpha}I_{\alpha}
-\frac1{2(1+t\epsilon)^2}\tr\left(
(K+L)^{-1}\Lambda
(K+L)^{-1}\Lambda
\right),
\end{align}
where the notations above are explained as follows:
$\Lambda:=[(1+\delta_{ij})\lambda_{e_i+e_j}]_{i,j}$, with $e_i\in\{0,1,\dots\}^d$ being the vector with only the $i$-th coordinate equal to 1 and the other $d-1$ coordinates equal to 0;
moreover for each multi-index $\alpha=(\alpha_1,\dots,\alpha_d)\in\{0,1,\dots\}^d$ we defined 
\begin{align}
I_{\alpha}
=
-\frac1{1+t\epsilon}\lambda_{\alpha}^2\prod_{i=1}^d\frac{\alpha_i!}{k_i^{\alpha_i}}
+\frac{1+u}{1-t\epsilon}
\lambda_{\alpha}^2\prod_{i=1}^d\frac{\alpha_i!}{(k_i+u)^{\alpha_i}}.
\end{align}
For each multi-index $\alpha$ (viewed as a vector in $\mathbb{R}^d$) we define its norm $|\alpha|$ as the sum of its coordinates.
Then since $\int r=0$ and $\int xr=0$,
for $|\alpha|\le 1$ we must have $\lambda_{\alpha}=0$ and hence $I_{\alpha}=0$.
Now we claim that if  $\epsilon$ is chosen small enough (the choice depending only on $u,K,L$ and not on $(\lambda_{\alpha})$), we have
\begin{align}
\sum_{\alpha\colon |\alpha|=2}I_{\alpha}
-\frac1{2(1+t\epsilon)^2}\tr\left(
(K+L)^{-1}\Lambda
(K+L)^{-1}\Lambda
\right)\le 0
\label{e161}
\end{align}
and 
\begin{align}
I_\alpha\le 0, \quad\forall 
\alpha\colon|\alpha|\ge 3.
\label{e162}
\end{align}
Clearly these claims would establish $\frac{d^2}{dt^2}\Theta(p_t,L)\le 0$ for $t\in[0,1]$ and hence the conclusion of the theorem.

To see \eqref{e161}, it suffices to show that
\begin{align}
\frac{(1+\epsilon)^2}{1-\epsilon}
\le 
\inf_{(\lambda_{\alpha})}\frac{\sum_{|\alpha|=2}\lambda_{\alpha}^2\prod_{i=1}^d\frac{\alpha_i!}{k_i^{\alpha_i}}+\frac1{2}\tr((K+L)^{-1}\Lambda(K+L)^{-1}\Lambda)}{(1+u)\sum_{|\alpha|=2}\lambda_{\alpha}^2\prod_{i=1}^d\frac{\alpha_i!}{(k_i+u)^{\alpha_i}}}
\label{e163}
\end{align}
where the infimum is over all $(\lambda_{\alpha})_{\alpha\colon|\alpha|=2}$ such that the denominator in \eqref{e163} is positive.
Since both the numerator and the denominator in \eqref{e163} are 2-homogeneous in $(\lambda_{\alpha})_{\alpha\colon|\alpha|=2}$ (note that $\Lambda$ is also a function of $(\lambda_{\alpha})_{\alpha\colon|\alpha|=2}$), the infimum can be restricted to a compact set, and hence it can be achieved by some $(\lambda_{\alpha})_{\alpha\colon|\alpha|=2}$.
By simple calculations we can show that the strict negativity of the Hessian in the Gaussian optimization problem \eqref{e163}
is in fact equivalent to the right side of the infimum in \eqref{e163} being strictly larger than 1 for each $(\lambda_{\alpha})$ (to see this, consider a perturbation of $K$ in the direction of $\Lambda$ and compute the second derivatives). 
Therefore \eqref{e163} and hence \eqref{e161} is true for all $\epsilon\le \epsilon_1$ where $\epsilon_1>0$ is some constant depending only on $u, K$, and $L$.

To see \eqref{e162}, it suffices to show that 
\begin{align}
\frac{1-\epsilon}{1+\epsilon}
\ge (1+u)\prod_{i=1}^d\left(\frac{k_i}{k_i+u}\right)^{\alpha_i},\quad\forall \alpha\colon|\alpha|=3.
\label{e164}
\end{align}
The assumption of $k_i<\frac{u}{(1+u)^{1/3}-1}$ implies that the right side of \eqref{e164} is strictly less than 1, and hence \eqref{e164} holds for all  $\epsilon\le \epsilon_2$ where $\epsilon_2$ depends only on $u,K,L$. The claim of the theorem then follows by taking $\epsilon=\min\{\epsilon_1,\epsilon_2\}$.
\end{proof}

\section{A Geometric Inequality Analogous to Lemma~\ref{lem2}}\label{sec_geometric}
In this section 
we discuss an analogous convex geometric inequality with a non-isotropic extremal, which may be of independent interest.
The similarities between entropic and geometric inequalities are well-known (e.g.\ \cite{dembo1991information}\cite{aras2021sharp})
and sometimes can be employed to prove new results in network information theory \cite{liu2022minoration}.
Intuitively, differential entropy can be seen as analogous to the logarithmic volume of a set,
and the sum of random variables is analogous to the Minkowski sum of sets.
Therefore, we may consider an analogue of the inequality in Conjecture~\ref{conj1} by replacing the entropy of independent sum with the log volume of the Minkowski sum.

Recall that the meanwidth of a convex body $C\subseteq\mathbb{R}^d$ is defined by 
\begin{align}
W(C):=\mathbb{E}[\sup_{u\in C}\left<u,X\right>-\inf_{u\in C}\left<u,X\right>]
\end{align}
where $X$ has a uniform distribution on the unit sphere in $\mathbb{R}^d$.
We will replace the power constraint in the entropic inequality by a meanwidth constraint for a convex body.
\begin{theorem}\label{thm2}
Let $B$ be the ball of radius $1/2$ in $\mathbb{R}^d$ (so that its meanwidth equals 1).
The maximum of 
\begin{align}
\frac{\vol^{1/2}(K+B+L)\vol^{1/2}(K)}{\vol(K+B)}
\label{e1}
\end{align}
over all convex body $K$ and all convex body $L$ satisfying the meanwidth bound $W(L)\le 1$ is strictly larger than 1. 
On the other hand, if we restrict $L=B$, then the supremum of the same quantity equals 1.
\end{theorem}
\begin{proof}
We first prove the claim for $L=B$. Note that by the Brunn-Minkowski inequality ($\vol(\alpha A+\beta B)\ge \vol^{\alpha}(A)\vol^{\beta}(B)$), we see that \eqref{e1} is bounded above by 1.
Moreover, \eqref{e1} approaches $1$ if we take $K=tK$ where $K$ is any convex body and $t\to\infty$.

Next, we show that the supremum can be strictly larger than 1 without the restriction that $L=B$.
Let us consider the case of $k=2$ for simplicity. (See Remark~\ref{rem1} for more general settings.)
Let $K$ be the unit cube (similar idea actually works for any convex body $K$).
Let $L$ be $\frac{\pi}{4}K$ rotated by $\pi/4$.
Note that the meanwidths
\begin{align}
W(L)=W(B)=1.
\end{align}
However, for $\epsilon$ small we have
\begin{align}
\vol(K+\epsilon L)&=1+4\epsilon\cdot\frac{\pi\sqrt{2}}{4}+O(\epsilon);
\\
\vol(K+\epsilon B)&=1+4\epsilon +O(\epsilon^2).
\end{align}
Therefore $\vol(K+\epsilon L)>\vol(K+\epsilon B)$ for small enough $\epsilon>0$.
Now let $K=tK$ and let $t\to\infty$. We have
\begin{align}
\frac{\vol(K+B)}{\vol(K)}&=1+4t^{-1}+O(t^{-2});
\\
\frac{\vol(K+B+L)}{\vol(K+B)}&=1+\pi\sqrt{2}t^{-1}+O(t^{-2}).
\end{align}
Therefore \eqref{e1} equals
\begin{align}
1+    \frac{\pi\sqrt{2}-4}{2}t^{-1}+O(t^{-2})
\end{align}
which is positive for $t$ large enough.
\end{proof}
\begin{remark}\label{rem1}
More generally, consider the supremum of 
\begin{align}
\frac{\vol^u(K+B+L)\vol(K)}{\vol^{1+u}(K+B)}
\end{align}
where $u>0$, and suppose that the convex bodies are in $\mathbb{R}^d$.
In the above argument we can consider $K$ to be a cube of growing size and let $L$ be the cross-polytope with diameter of order $\sqrt{\frac{d}{\ln d}}$ (so that the mean-width is order of a constant).
Then
\begin{align}
\frac{\vol(K+B)}{\vol(K)}&=1+\Theta(dt^{-1});
\\
\frac{\vol(K+B+L)}{\vol(K+B)}&\approx
\frac{\vol(K+L)}{\vol(K)}
\\
&=1+\Theta(dt^{-1}\sqrt{\frac{d}{\ln d}}).
\end{align}
Therefore the supremum is positive if $dt^{-1}<<udt^{-1}\sqrt{\frac{d}{\ln d}}$,
or equivalently $u>>\sqrt{\frac{\ln d}{d}}$.
On the other hand, the supremum over ball $L$ is 0 iff $u\le 1$. 
Therefore in the regime of 
$u$ between $1$ and $\omega(\sqrt{\frac{\ln d}{d}})$ the supremum is not achieved by balls.
\end{remark}

\section{Discussion}
\label{sec_discussion}
It seems that simple amendments of Conjecture~\ref{conj1}, such as taking the concave envelope in $X_2$ or restricting to parameters in the stability regime, no longer fulfill its original purpose of implying Gaussian optimality for the HK region.
However, it is still possible to imply Gaussian optimality for some parts of the region.
For example, the special case of Conjecture~\ref{conj1} in \cite{beigi2016some} concerning Costa's corner point still appears to be a valid approach towards the slope of the corner point.
Here we slightly generalize the observation therein to a larger part of the rate region.
\begin{theorem}\label{thm8}
Suppose that $N_1,N_2,u,q_1,q_2\in(0,\infty)$ are such that the maximum of
\begin{align}
u h(X_1+X_2+Z_1+Z_2)
+h(X_1+Z_1)-
(1+u) h(X_1+Z_1+Z_2)
\end{align}
over independent one dimensional random variables $X_1,X_2\in\mathbb{R}^d$,
$\mathbb{E}[\|X_i\|_2^2]\le dq_i$, $i=1,2$ is achieved by isotropic Gaussian distributions (possibly degenerate with zero covariance) for any positive integer $d$. 
Then the HK bound \eqref{e120} with Gaussian inputs is tight for the weighted sum rate $R_1+(1+u)R_2$.
\end{theorem}
\begin{proof}
The proof is essentially the same as \cite[Lemma~4]{beigi2016some};
we can split the optimization problem into two, which achieve the supremum simultaneously:
\begin{align}
&\quad R_1+(1+u)R_2
\nonumber\\
&\le 
\lim_{d\to \infty}\frac1{d}
\sup
\left\{I(X_1;Y_1)+(1+u)I(X_2;Y_2)
\right\}
\label{e200}
\\
&\le \lim_{d\to \infty}\frac1{d}
\left\{\sup h(Y_2)
+\sup[uh(Y_2)-(1+u)h(Y_2|X_2)+h(Y_1)]-h(Y_1|X_1)\right\}
\label{e201}
\\
&\le \frac1{2}\ln\frac{q_1+q_2+N_1+N_2}{N_1}+\frac1{2}\sup_{K\le q_1}\psi(K,q_2)
\label{e202}
\end{align}
where the suprema are over independent $X_1,X_2\in\mathbb{R}^d$ satisfying $\mathbb{E}[\|X_i\|_2^2]\le dq_i$, $i=1,2$; $Y_1$ and $Y_2$ are defined in \eqref{e_y1}-\eqref{e_y2};
$\psi(K,q_2)$ was defined in \eqref{e129} (with dimension one);
\eqref{e200} follows by Fano's inequality;
\eqref{e202} follows by the assumption of Theorem~\ref{thm8}.
Note that \eqref{e202} is upper bounded by the HK bound \eqref{e120} (in dimension one with Gaussian inputs), therefore the claim of the theorem follows.
\end{proof}
The set of $(N_1,N_2,u,q_1,q_2)$ satisfying the assumption in Theorem~\ref{thm8} always falls in the setting of Lemma~\ref{lem5}, 
and the latter always falls in the regime of stable Gaussian stationary points.
The set is not empty, as the sufficient condition in \cite{gohari2021information} shows.
A more precise characterization of this set remains an open question.
In particular, a concrete new conjecture can be formulated as follows:
\begin{conjecture}\label{conj2}
The set of $(N_1,N_2,u,q_1,q_2)$ satisfying the isotropic condition in Theorem~\ref{thm8} equals the set of $(N_1,N_2,u,q_1,q_2)$ for which 1-letter HKGS without power control matches 1-letter HKGS with power control.
\end{conjecture}
In Conjecture~\ref{conj2}, 1-letter HKGS means the Han-Kobayashi bound \eqref{e120} in dimension one with Gaussian signaling, and with/without power control means $(Q_1,Q_2)$ in \eqref{e120} is random/constant.
The first mentioned set in Conjecture~\ref{conj2} is contained in the second, by the proof of Theorem~\ref{thm8}. 
Proving Conjecture~\ref{conj2}
would prove the conjecture about the slope at Costa's point in \cite{beigi2016some}. 
On the other hand, disproving Conjecture~\ref{conj2} would disprove the optimality of HKGS with power control for the whole capacity region.

\section{Acknowledgement}
The author would like to thank Professors Ramon van Handel,
Chandra Nair, and Max Costa for inspiring discussions and comments on the manuscript,
which were helpful in improving the presentations and the formulation of Conjecture~\ref{conj2}.
This work was supported by the start-up grant at the Department of Statistics, University of Illinois at Urbana-Champaign.

\bibliographystyle{plainurl}
\bibliography{ref}
\end{document}